\documentclass[twoside]{article}
\usepackage{amsfonts}
\usepackage{stmaryrd}
\usepackage{amssymb}
\usepackage{euscript}
\usepackage{pstricks}
\usepackage{amsthm}
\usepackage{latexsym,amsmath,amscd}
\usepackage{mathrsfs}
\usepackage{graphicx}
\usepackage{color}
\usepackage{dsfont}
\usepackage{bm}
\usepackage{pgf,tikz}
\usetikzlibrary{arrows}

\newtheorem{theorem}{Theorem}[section]
\newtheorem{definition}{Definition}[section]
\newtheorem{proposition}{Proposition}[section]
\newtheorem{assumption}{Assumption}[section]
\newtheorem{lemma}{Lemma}[section]
\newtheorem{corollary}{Corollary}[section]
\newtheorem{remark}{Remark}[section]
\newtheorem{example}{Example}[section]

\numberwithin{equation}{section}

\newcommand{\Isi}{b^\si}
\newcommand{\wtIsi}{\wt{b}^\si}
\newcommand{\bbUU}{\bar{U}}

\renewcommand{\Game}{\mathcal{G}}

\newcommand{\wt }{\widetilde }
\newcommand{\wh }{\widehat }

\newcommand{\N}{\mathbb{N}}

\newcommand{\R}{\mathbb{R}}

\newcommand{\E}{\mathbb{E}}

\renewcommand{\H}{\mathbb{H}}
\renewcommand{\S}{\mathbb{S}}

\renewcommand{\P}{\mathbb{P}}

\newcommand{\FF}{\mathbb{F}}

\newcommand{\EMM}{ {\mathbb E}_{\mathbb{Q}}}

\newcommand{\Q}{\mathbb{Q}}

\newcommand{\Aset}{\mathcal{N}}

\newcommand{\Filt}{\mathcal{F}}

\newcommand{\Prob}{\mathbb{P}}
\newcommand{\Qrob}{\mathbb{Q}}
\newcommand{\Mset}{\mathcal{M}}
\newcommand{\Eset}{\mathcal{E}}
\newcommand{\Strat}{\mathcal{S}}

\newcommand{\STOP}{\mathcal{T}}

\newcommand{\Act}{\mathcal{A}}
\newcommand{\Hist}{\mathcal{H}}

\newcommand{\Kmot}{\mathcal{K}^\dagger}

\newcommand{\esssup}{\operatornamewithlimits{ess\,sup}}
\newcommand{\essinf}{\operatornamewithlimits{ess\,inf}}

\let\inf\relax \DeclareMathOperator*\inf{\vphantom{p}inf}
\let\essinf\relax \DeclareMathOperator*\essinf{\vphantom{p}ess\,inf}

\newcommand{\map}[2]{\,{:}\,#1\rightarrow#2}

\newcommand\I{\mathds{1}}

\newcommand{\si}{\sigma}

\newcommand{\Val}{Z}

\newcommand{\stp}{{\vartheta}}

\newcommand{\cll}{[\![}
\newcommand{\orr}{)\!)}
\newcommand{\crr}{]\!]}

\newcommand{\ZRG}{\operatorname{ZRG}} 
\newcommand{\MRG}{\operatorname{MRG}} 

\newcommand{\ASG}{\operatorname{ASG}} 

\newcommand{\MGCC}{\operatorname{MGC}} 

\renewcommand{\bm}{}

\newcommand{\Keywords}[1]{\par\noindent{\small{\bf Keywords\/}: #1}}
\newcommand{\Class}[1]{\par\noindent{\small{\bf Mathematics Subjects Classification (2010)\/}: #1}}

\title{{\Large \bf ARBITRAGE PRICING OF MULTI-PERSON GAME CONTINGENT CLAIMS}\vskip 90 pt}

\author{Ivan Guo and Marek Rutkowski\footnote{The research of Ivan Guo and Marek Rutkowski was supported under Australian Research Council's Discovery Projects funding scheme (DP120100895).}
\\ School of Mathematics and Statistics
\\ University of Sydney
\\ NSW 2006, Australia}

\date{\vskip 60 pt 6 April 2014 \vskip 25 pt}

\begin{document}
\maketitle
\vskip 45 pt
\begin{abstract}
We introduce a class of financial contracts involving several parties by extending the notion of a two-person game option to a contract in which an arbitrary number of parties is involved and each of them is allowed to make a wide array of decisions at any time, not restricted to simply `exercising the option'. The collection of decisions by all parties then determines the contract's settlement date as well as the terminal payoff for each party. We provide sufficient conditions under which a multi-person game option has a unique arbitrage price, which is additive with respect to any partition of the contract.
\vskip 20 pt
\vskip 20 pt
\Class{91A06,$\,$91A15,$\,$60G40}
\end{abstract}


\newpage

\section{Introduction}\label{chgeneralcontract}

In financial markets, almost all traded derivatives effectively involve only two parties, the issuer and the holder. Furthermore, in most cases, the only meaningful decision lies with the holder in the form of exercising the contract. Some exceptions to this include the two-person game options introduced by Kifer \cite{Kifer1}  in which the issuer also has decisions to make (see also Dolinsky and Kifer \cite{Dolinsky}, Dolinsky et al. \cite{Dolinsky11}, Kallsen and K\"uhn \cite{Kallsen1},
K\"uhn et al. \cite{Kuehn}, Kyprianou \cite{Kyprianou} and a review paper by Kifer \cite{Kifer2}). From the practical perspective, one may refer to convertible bonds as real-world financial derivatives covered by the general concept of the two-person game option (see, for instance, Ayache et al. \cite{Ayache}, Andersen and Buffum \cite{Andersen}, Bielecki et al. \cite{Bielecki}, or Kallsen and K\"uhn~\cite{Kallsen2}). No wonder that the existing mathematical theory of pricing financial derivatives is mainly concerned with bilateral deals. By contrast, our goal is to open a new avenue by designing an extended framework in which one may define and evaluate contingent claims involving several parties.

In a general version of a multi-person contingent claim introduced in this work, an arbitrary number of parties is involved and each of them is allowed to make a wide array of decisions at any time, not restricted to simply `exercising the option'. Decisions of all parties then determine the settlement time of the claim, as well as their respective terminal payoffs. Since one may draw equivalences between the mechanisms of multi-person financial contracts and general stochastic games examined in Guo and Rutkowski \cite{Guo1,Guo2,Guo3}, we propose the name \emph{multi-person game contingent claims} or \emph{multi-person game options}. We work throughout within the discrete or continuous-time arbitrage-free market model with the intention of adding multi-person game contingent claims to this market. The generalised framework proposed here is not restricted to one particular issue, but has instead several purposes. First, it allows for the creation and valuation of new complex derivatives involving multiple parties. Second, it provides fresh insights to the valuation of existing exotic derivatives. Finally, the framework may serve as a starting point to a unified approach of modelling financial derivatives in the presence of external factors, such as: credit risk, funding costs, and margin account (see, for instance, \cite{Bielecki1,Pallavicini}). For instance, one may ask how the fact that external funders can make meaningful decisions in their own interest may affect the valuation of a contract between either two or several counterparties.

The paper is organized as follows. In Section \ref{sec7d}, we provide the definition of a multi-person game contingent claim, which is associated with a multi-player stochastic game of perfect information. The claim involves several \emph{holders} of parts of a contract (dubbed \emph{tranches}), as well as its \emph{issuer}. Since each party is able to observe the actions of others, this information is reflected in their decisions and hedging portfolios. As a result, admissible trading strategies are redefined as \emph{$\Hist$-predictable trading strategies}, which also involve reactions to the observable actions of all parties. An \emph{enter-and-hold arbitrage} is defined in a natural way as the opportunity to make a guaranteed profit by holding tranches along with an $\Hist$-predictable trading strategy until the claim is settled.

Section \ref{sec7da} focusses on the pricing of a particular \emph{combined tranche,} that is, a fixed collection of tranches.
To establish the no-arbitrage bounds for a combined tranche, it suffices to consider a holder of that tranche and
to postulate that the issuer is holding all remaining tranches of a contract.
 The main results in this section are Theorems \ref{thmro10} and \ref{thmsg30}, for the discrete and continuous-time case, respectively. They show that, under some sensible conditions on payoffs, no enter-and-hold arbitrage is possible for a fixed combined tranche if and only if its price lies within a particular interval. The bounds of this interval are then interpreted in terms of super-hedging strategies.

In Section \ref{sec7e2}, we consider in addition the interaction between different combined tranches. In particular, we demonstrate that prices of single tranches of a multi-person game option must be additive in order to avoid a second type of arbitrage, dubbed the \emph{immediate reselling arbitrage}. We show that the additivity of prices also eliminates any price disagreement between the parties.
Theorem \ref{thmry30} furnishes the necessary and sufficient conditions for the non-existence of both types of arbitrage opportunities. Theorem \ref{thmrx10} proves that if optimal equilibria exist for the associated multi-player game, then all individual tranches of a game contract have unique arbitrage prices, which match the values of the game. Furthermore, if the payoffs satisfy a sub-zero-sum condition, then every combined tranche has a unique arbitrage price, which is also additive. For the reader's convenience, the pertinent definitions and results from the game theory are collected in the appendix.

\subsection{Multi-Person Contract with Puttable Tranches} \label{mutyy}

For concreteness, let us consider a continuous-time complete market model of a single stock, e.g., the Black-Scholes model.
Our goal is to examine arbitrage pricing of contracts involving several parties. One of
them, called the \emph{issuer}, plays a special role -- he initiates (`sells')
the contract and takes care of its settlement according to the rules specified
below. All other $m$ parties are referred to as the {\it holders} of {\it tranches} of a contract.
We give below an example of a simple contract between the issuer and $m$ holders.
In Section \ref{ssc4.5}, we show how to value this contract using no-arbitrage principles.

\subsubsection{Two-Period Setup}

We consider dates $0<T_1< T_2 =T$ and we assume that
at time 0 each holder buys a tranche, formally represented by a call or put option of
either European or American style. Options may have different strikes, but they share the same expiry date $T$,
and they cannot be exercised before the date $T_1$. At time $T_1$, each holder has the right to give back (that is, to {\it put}) the tranche to the issuer. If the $i$th holder decides to put the tranche he owns, then he receives from the issuer an $\Filt_{T_1}$-measurable amount $X^i_1$ specified in the contract. We denote the set of such holders by $\Eset$.
 If all holders decide to put the tranche  to the issuer at time $T_1$, then he simply pays all amounts $X^1_1, \dots , X^m_1$ to holders. However, if at least one holder decides to keep his tranche,
then the contract stipulates that the issuer computes his {\it total deviation} (the `total loss') at time $T_1$, as given
by the sum $\sum_{i\in\Eset}(X^i_1 - P^i_1)$ where $P^i_1$ is the price of the $i$th option
at time $T_1$, and the total loss of the issuer is covered by the holders who decided to
keep their tranches.

 For simplicity, we assume that the \emph{redistribution} is split evenly amongst the holders who decided to keep the tranche, so that each of them should pay the issuer the amount of $\frac{1}{k}\sum_{i\in\Eset}(X^i_1 - P^i_1)$ where $k=m-|\Eset|$ is the number of such holders, although more general redistribution rules can also be implemented
(see \cite{Guo1,Guo2,Guo3}). We note that it is possible to choose strikes $K_1, \dots , K_m$, classes of options (calls or puts),
and the amounts $X^1_1, \dots , X^m_1$ to ensure that it will be not `optimal' for all holders
to put their tranches to the issuer at time $T_1$. Specifically, the amounts $X^1_1, \dots , X^m_1$
may be chosen small enough so that there exists at least one holder for which it would be suboptimal
to put his tranche, even when all other holders would prefer to do so. It is natural to conjecture that in this
situation the issuer's `total valuation' problem will reduce to the valuation of individual options, since
holders' decisions at time $T_1$ will not affect him. This means that the sum of prices of `tranches' at time 0
should be equal to the sum of prices of options. However, the price at time 0 of a tranche does not necessarily coincide with
the price of the corresponding option, due to the redistribution of payoffs at time $T_1$.

\subsubsection{Multi-Period Setup}

Let us now consider the dates $0<T_1<T_2<\cdots<T_n< T_{n+1}=T$.
Each holder buys a tranche represented by an option expiring at $T$,
which cannot be exercised before $T_n$. To describe the mechanism of the contract, let us fix some $l=1,2,\ldots,n$. At time $T_l$, each current holder has the right to put the tranche to the issuer. If the $i$th holder decides to do so, then he receives a predetermined $\Filt_{T_l}$-measurable amount $X^i_l$. We denote the set of such holders by $\Eset_l$.
If all tranches are given back to the issuer at time $T_l$, then he pays all amounts $X^1_l, \dots , X^m_l$ to holders.
If at least one holder decides to keep his tranche, then the issuer computes his `total loss' at time $T_l$,
which is given by the sum  $\sum_{i\in\Eset_l}(X^i_l - P^i_l)$ where $P^i_l$ is the \emph{market price}
(\emph{continuation value}) of the $i$th tranche at time $T_l$. The price $P^i_l$ of tranche $i$ is computed
assuming that no putting of any tranche occurs at $T_l$ (in other words, before holders decisions at time $T_l$ were
announced).  The total loss of the issuer is covered by the holders who decided to keep their tranches.
The redistribution is split evenly amongst the holders who decided to keep the tranche at $T_l$,
meaning that each of them compensate the issuer by paying him at time $T_l$ the amount of $\frac{1}{k}\sum_{i\in\Eset_l}(X^i_l - P^i_l)$ where $k=m-|\Eset_l|$ is the number of such holders. After the redistribution stage has been completed, we move on to the next date $T_{l+1}$, with the issuer being the new holder of tranche $i$ for all $i\in\Eset_l$. Then he can make any decisions available to the holder of tranche $i$ from time $T_{l+1}$ onwards. He can also sell some or all of these tranches to new holders and thus relinquish the corresponding decision-making abilities. We repeat these steps until the decisions and redistributions at date $T_n$ are completed. At this moment, the tranches revert back to original options.

\section{Multi-Person Game Options}\label{sec7d}

All of the following definitions refer to an underlying probability space $(\Omega, \Filt, \P )$, which is endowed with the filtration $\FF = \{\Filt_t: t\in[0,T]\}$ representing the information flow observed by all market participants.
In continuous-time setup, the filtration $\FF$ is also assumed to satisfy the `usual conditions' that underpin stochastic calculus.
We denote by $B, S^1 , \dots , S^d$ the $\FF$-adapted processes representing prices of primary traded assets.
For simplicity of presentation, it is postulated throughout that the underlying market model (which is left largely unspecified)
is frictionless, arbitrage-free, and complete. Hence, according to the classic version of FTAP, the unique martingale measure, denoted by $\Q$, is available. For a detailed description of discrete-time case (continuous-time case, resp.), we refer to Jacod and Shiryaev \cite{Jacod} (Jeanblanc et al. \cite{Jeanblanc}, resp.). Further technical assumptions about the continuous time model will be stated in Section \ref{sec7f}.

\subsection{Actions, Outcomes and Strategy Profiles}\label{psec7d}

Consider an arbitrary game $\Game$ with player set given by $\Mset=\{1,2,\ldots,m\}$ and the space of \emph{actions} $\Act = \prod_{i\in\Mset} \Act^i$. We assumed that each $\Act^i$ is a finite (or infinite, but countable) set representing actions available to player $i$ at any time. Unless otherwise stated, the definitions cover both discrete- and continuous-time cases.
In particular, $[t,u]$ denotes the set $\{t,  \dots , u \}$ when a discrete-time framework is studied.
Note also that the convention $[0,t)=[0,t-1]:= \{0, \dots , t-1\}$ is used in that case for $t=1,2, \dots ,T$.

\begin{definition} \label{outcome} {\rm
An \emph{outcome} for party $i$ is an $\mathbb{F}$-adapted process $h^i = (h^i_t:t\in[0,T])$ with values in $\Act^i$.
It represents the actions of party $i$ over time.  The space of outcomes for party~$i$ is denoted by $\Hist^i$. We write
$h=(h^1,\ldots,h^m)$ and $\Hist =\prod_{i\in\Mset} \Hist^i$ to denote the $m$-tuple of outcomes and its space, respectively.}
\end{definition}

We will frequently consider a sequence of actions up to, and possibly including, a certain time $t\leq T$.
For any outcome $h\in\Hist$, the \emph{history before time} $t$ is the restriction of $h$ to the time interval $[0,t)$.
We denote it by $h_{[0,t)}$ and we write  $\Hist_{[0,t)}$ to denote the space of such histories.
Also, we define $h_{[0,t]}$ and $\Hist_{[0,t]}$ in a similar way.
The following technical assumption is relevant in the continuous-time case only, in which case it will be indeed crucial.

\begin{assumption}\label{assdv01} {\rm
Every outcome $h^i\in\Hist^i$ for every $i\in\Mset$ is a right-continuous process.}
\end{assumption}

Many aspects of a game only depend on the past history only. To quantify this notion,
let us consider a mapping $f\map{\Hist\times[0,T]\times\Omega}{C}$ with values in some space $C$.
For any $h,h'\in \Hist$, we define the random time
\begin{gather}  \label{veqsg40}
\rho (h,h') := \inf \big\{ t\in[0,T]: h_t \neq h'_t \big\} \wedge T.
\end{gather}
Then $\rho (h,h') $ is an $\FF$-stopping time, since the  
outcomes $h$ and $h'$ are $\FF$-adapted. 
Obviously, for every $t\in [0,T]$, the event $\{h_{[0,t)}=h'_{[0,t)}\}$ coincides with $\{t \leq \rho (h,h') \}$.

\begin{definition} \label{defpredi} {\rm
A mapping $f\map{\Hist\times[0,T]\times\Omega}{C}$ is \emph{$\Hist$-adapted} (\emph{$\Hist$-predictable}, resp.) if for every $h,h'\in\Hist$, the equality $f(h)=f(h')$ holds on the stochastic interval $\cll 0, \rho \orr $ (on $\cll 0, \rho \crr$, resp.)
where $\rho = \rho (h,h') $.}
\end{definition}

It is clear that if a mapping $f$ is $\Hist$-predictable then it is also $\Hist$-adapted.
Intuitively, a mapping $f$ is $\Hist$-adapted ($\Hist$-predictable, resp.) if, for all $t \in [0,T]$,
the value $f_t(h)$ depends only on $h_{[0,t]}$ (on $h_{[0,t)}$, resp.).
Informally, for an $\Hist$-adapted ($\Hist$-predictable, resp.) mapping $f$ one could thus write
$f_t (h)= f_t (h_{[0,t]})$ ($f_t (h)= f_t (h_{[0,t)})$, resp.) for every $h \in \Hist$  and $t \in [0,T]$.

A \emph{strategy} for player $i$ specifies how she should react to the observable actions of other parties.
It is mapping from the space $\Hist $ of outcomes to the set $\Act^i$ of actions, which also takes
into account the information flow represented by the filtration $\FF$.

\begin{definition} \label{qazqaz} {\rm
A \emph{strategy} of party $i$ is an $\FF$-adapted and $\Hist$-predictable mapping $s^i : \Hist\times[0,T]\times \Omega \to \Act^i$. The $m$-tuple of strategies $s=(s^1,\ldots,s^m)$ is called a \emph{strategy profile}. We denote by $\Strat^i$  ($\Strat$, resp.) the space of all possible strategies of party $i$ (strategy profiles, resp.).}
\end{definition}

We define the mappings $s_{[0,t)}\in \Strat_{[0,t)}$ and $s^i_{[0,t)}\in\Strat^i_{[0,t)}$ in an obvious way. Note also that wherever the superscript $i\in\Mset$ appears, it can be replaced by an arbitrary subset $\Aset\subseteq\Mset$ to describe the corresponding notion relating to $\Aset$. Moreover, $-\Aset$ is a shorthand notation for $\Mset\setminus\Aset$.
It will be frequently convenient to write
\begin{equation} \label{defsigg}
s = (\tau ,\si) = (s^{\Aset }, s^{-\Aset }) \in\Strat^{\Aset}\times\Strat^{-\Aset}.
\end{equation}

\begin{definition} {\rm
For any strategy profile $s\in\Strat$, we denote by $h(s) \in \Hist $ the \emph{outcome associated} to $s$, that is, the unique outcome obtained when a strategy profile $s$ is played on $[0,T]$.}
\end{definition}

It is natural to define a game by first specifying the class $\Strat^i$ of all strategies available to each player
and, subsequently, to define $\Hist^i$ and $\Hist $ as the sets of all outcomes associated with all possible strategy
profiles $s \in \Strat$. This specification underpins all definitions and results stated in what follows.

Let us now consider the game at time $t$. We assume that at any date $t$ the sample path of the history $h_{[0,t)}(s)$ (but not
the mapping $s_{[0,t)}$ used to generate this sample path) is observed by all parties. For any fixed $s = (\tau,\si) \in\Strat$ and $\Aset\subseteq\Mset$, we define the set $\Strat^\Aset_t (s) $ of strategies available to players from the set $\Aset$ at time $t$ after a strategy profile $s_{[0,t)}\in\Strat_{[0,t)}$ was played. Formally, the subset $\Strat^\Aset_t(\tau,\si) \subseteq\Strat^\Aset$ is given by
\begin{equation} \label{defstr}
\Strat^\Aset_t(\tau,\si) := \big\{ \tau' \in\Strat^\Aset : h_{[0,t)}(\tau',\si) =h_{[0,t)}(\tau,\si) \big\} .
\end{equation}

\subsection{Multi-Person Game Contingent Claims}\label{bsec7d}

We first formulate a general definition of a multi-person game contract. An important ingredient of a contract is a decision to terminate it, which may be chosen by any party $i$ provided, of course, that this possibility belongs to the action set $\Act^i$. In an abstract formulation of a game claim, it will be enough to introduce the concept of the \emph{termination mapping} through the following definition.

\begin{assumption} \label{assmm09} {\rm
The \emph{termination mapping} $\Theta \map{\Hist \times [0,T] \times \Omega}{\{0,1\}}$ satisfies: \hfill \break
(i) for every $h \in \Hist$, the process $\Theta(h)$ is $\FF$-adapted, non-decreasing and RCLL, \hfill \break
(ii) the mapping $\Theta $ is $\Hist$-adapted.}
\end{assumption}

For any $h \in \Hist $, the termination date is defined by
\begin{equation} \label{defter}
\stp (h) = \inf \{ t \in [0,T]: \Theta_t(h)=1 \}.
\end{equation}
Note that $\stp(h) \in \STOP_{[0,T]}$ for every $h \in \Hist $, where $\STOP_{[0,T]}$ stand for the class of all $\FF$-stopping times with values in $[0,T]$. Condition (ii) in Assumption \ref{assmm09} means that for every $h,h' \in \Hist $, if $h_{[0,t]} =h'_{[0,t]}$ and $\stp (h) = t$ then $\stp (h') = t$. Given a strategy profile $s$ and the associated outcome $h(s)$, we denote by $\stp (s)$ the termination date $\stp (h(s))$. Let $V_{\stp (s)}(s) = ( V^1_{\stp (s)}(s), \dots, V^m_{\stp (s)}(s))$
stand for the vector of random payoffs, which are received by holders at time $\stp(s)$ if a strategy profile $s$ is played or,
equivalently, if the outcome $h(s)$ is realised. Each particular payoff $V^i_{\stp (s)}(s)$ will be identified with a single \emph{tranche} of a contract.

\begin{assumption} \label{aspmm09} {\rm
The \emph{payoff mapping} $V \map{\Hist \times [0,T] \times \Omega}{\mathbb{R}^m}$ satisfies: \hfill \break
(i) for every $h \in \Hist$, the processes $V^1(h), \dots, V^m(h)$ are $\FF$-adapted and RCLL, \hfill \break
(ii) the mapping $V$ is $\Hist$-adapted.}
\end{assumption}

In the classic case of a two-person game option, as defined in Kifer \cite{Kifer1}, one deals with two tranches, which are held by a \emph{holder} and an \emph{issuer}, respectively. Note that, despite being attributed different names, their roles and rights are here identical. It should thus be stressed that in a multi-person game claim the role of the \emph{issuer}  will be essentially different from that of ordinary holders.

We use the symbol $\Game $ to denote a generic game that underpins the specification of a multi-person game contingent claim. For explicitly defined classes of multi-player stochastic stopping games and results on the existence of an \emph{optimal equilibrium}, we refer to \cite{Guo1,Guo2,Guo3}. The definition and basic properties of an optimal equilibrium are also given in the appendix.

\begin{definition}\label{defrm01} {\rm
A \emph{multi-person game contingent claim} $\MGCC(\Game)$ is a contract on $[0,T]$, which consists of $m$ \emph{tranches}, between an \emph{issuer} and up to $m$ \emph{holders}. At time 0, a prospective holder of tranche $i$ has to pay the issuer some fee
(possibly negative) to enter the contract; we denote this amount by $\pi_0^i(\Game)$. Any individual party, including the issuer, is allowed to hold any combination of tranches. Once the contract starts, the holder of tranche $i$ may choose a strategy $s^i\in\Strat^i$, and decisions of all holders form a strategy profile $s=(s^1,\ldots,s^m)$. The contract is settled at time $\stp(s)\in\STOP_{[0,T]}$ where the issuer pays the holder of tranche $i$ the amount of $V^i_{\stp (s)}(s)$ for all $i=1,\ldots,m$.}
\end{definition}

In the real-world financial contracts, all cashflows are zero-sum in nature. Since a game $\Game$ is not postulated to be zero-sum, the crucial role of the issuer in $\MGCC(\Game)$ is to absorb any surpluses or cover any deficits at times 0 and $\stp(s)$. Consequently, the net cashflow between the holders of the $m$ tranches and the issuer is always zero.
If $\Game$ happens to be zero-sum, meaning that $\sum_{i\in\Mset} V_{\stp (s)} (s)=0$ for all $s\in\Strat$, then the issuer is somewhat redundant, since the holders can strike the deal between themselves at time 0 and then also settle the payoffs between themselves at time $\stp(s)$ without issuer's intervention. In practice, the role of an issuers is then reduced to a clearing house or intermediation.

\begin{remark} {\rm
One may also reinterpret the issuer as the holder of an \emph{issuer's tranche}, so that $\MGCC(\Game)$ becomes a zero-sum game contract with $m+1$ tranches in which the holder of the $(m+1)$th tranche is not allowed to take any actions. This corresponds to the concept of a \emph{dummy player} in \cite{Guo3}. Formally, the analysis of a game would then reduce to the zero-sum case.}
\end{remark}

\begin{remark} {\rm
The contract introduced in Definition \ref{defrm01} extends, in particular, the classic game option introduced in Kifer \cite{Kifer1}, which corresponds to a Dynkin stopping game between two players. The `holder' owns the tranche, which allows for `exercising', whereas the second holder (dubbed the `issuer') keeps the tranche allowing for `cancellation'. In the classic game options, there is no need to introduce an issuer in the sense of Definition \ref{defrm01}, since the classic game option is associated with a zero-sum stopping game. Aside from the number and roles of parties involved, we also generalise the two-person game option by allowing a greater complexity in the underlying game $\Game$. In particular, the analysis presented in what follows is by no means restricted to the case of stopping games.}
\end{remark}

\subsection{Adaptive Trading Strategies}\label{ssec7d}

Let $\MGCC_t(\Game,s)$ denote the multi-person  game contingent claim on $[t,T]$ after the strategy profile $s_{[0,t)}\in\Strat_{[0,t)}$ was played during $[0,t)$. Strictly speaking, $\MGCC_t(\Game,s)$ only depends on the history $h_{[0,t)}(s)$, meaning that if the equality $h_{[0,t)}(s)=h_{[0,t)}(s')$ holds, where $s'$ is any strategy profile, then the game claims $\MGCC_t(\Game,s)$ and $\MGCC_t(\Game,s')$ are in fact identical. To simplify presentation, we implicitly work under the event $\{\stp(s)\geq t\}$ when dealing with $\MGCC_t(\Game,s)$. In other words, we only consider the case where $\MGCC_t(\Game,s)$ has not yet been terminated before time $t$.

Denote each \emph{single tranche} $i\in\Mset$ by $\MGCC^i(\Game)$ and the collection of single tranches $\Aset\subseteq\Mset$, dubbed a \emph{combined tranche} $\Aset$, by $\MGCC^\Aset(\Game)$. We will simply use the term \emph{tranche} when no ambiguity may arise. For each $i\in\Mset$ and $\Aset\subseteq\Mset$, let $\MGCC^i_t(\Game,s)$ and $\MGCC^\Aset_t(\Game,s)$ denote the tranche  $i$ and the tranche  $\Aset$ at time $t$, respectively.

Let us now describe trading arrangements after the contract's initiation.
At any time $t$ before the settlement time, all tranches of $\MGCC_t(\Game,s)$ may be freely traded, so that the holder of the tranche  $i$ may sell it to a new holder. The new holder then continues the game with the other current holders and has the ability to choose a strategy $\tau \in\Strat^i_t(s)$ from time $t$ onwards. Denote the time $t$ price of the tranche  $i$
(tranche $\Aset$, resp.) by $\pi^i_t(\Game, s)$ ($\pi^\Aset_t(\Game, s)$, resp.). Let us stress that anyone, including the issuer, is allowed to hold any number of tranches at any time.

The pricing of game options hinges on no-arbitrage arguments and the replication (or super-hedging) of payoffs. In the current framework, every party (the issuer and holders) may hold a portfolio of stock and bonds in order to hedge her future liabilities. Since $\Game$ is a perfect information game, in which actions exist before $\stp(s)$ and are observed by all parties, it is natural for the issuer and the holders to adapt their trading strategies according to their current and past observations. In particular, if a strategy profile $s\in\Strat$ is played, a trading strategy should have a non-anticipative dependency on the observable outcome $h(s)\in\Hist$. This leads to the following definition in which it is implicitly assumed that the class of \emph{admissible}
 (thus, in particular, self-financing) trading strategies for a given market model $(B,S)$ was already specified (for instance,
 by focussing on trading strategies for which the discounted wealth is a $(\Qrob , \FF)$-martingale).

\begin{definition}\label{defrm02} {\rm
An \emph{$\Hist$-predictable trading strategy}  is a mapping $\phi\map{\Hist \times [0,T]\times \Omega}{\R^{d+1}}$ satisfying the following conditions: \hfill \break
\noindent (i) for every $h \in \Hist$, $\phi(h)$ is an admissible trading strategy,
\\ (ii) $\phi $ is an $\Hist$-predictable mapping.}
\end{definition}

The class of all $\Hist$-predictable trading strategies $\phi $ is denoted by $\Phi(\Hist)$.
For any $\phi \in \Phi(\Hist)$, the corresponding wealth mapping $\Val \map{\Hist \times [0,T]\times \Omega}{\R}$, which is given by the equality
\[
\Val_t(\phi(h))=\sum_{l=0}^d \phi^l_t(h) S^l_t,\quad \forall \, t\in[0,T],
\]
is also $\Hist$-predictable.  We will need to extend Definition \ref{defrm02} a little further. For a fixed $\si\in\Strat^{-\Aset}$, we denote by $\Hist^\si$ the set of all outcomes reachable when a strategy $\si$ is played, that is,
\begin{equation} \label{defhisi}
\Hist^\si := \big\{h\in\Hist : h=h(\tau,\si),\ \tau\in\Strat^\Aset\big\}.
\end{equation}
In some situations, where possible outcomes are assumed to belong to $\Hist^\si$, we may be only interested in trading strategies with the same constraint and restricted to the interval $[t,T]$.

\begin{definition}\label{devfrm02} {\rm
For any $\si\in\Strat^{-\Aset}$ and $t \in [0,T]$, we denote by $\Phi_t (\Hist^\si)$ the class of all
{\it $\Hist^\si$-predictable trading strategies} $\phi\map{\Hist^\si \times [t,T]\times \Omega}{\R^{d+1}}$.
The classes $\Hist^\tau$ and $\Phi_t (\Hist^\tau)$ are defined in an analogous manner.}
\end{definition}

Intuitively, Definition \ref{devfrm02} has the following interpretation: when a strategy profile $\si $ is fixed then, for every $u \in[t,T]$ and $\tau,\tau'\in\Strat^\Aset$, the equality $h_{[0,u)}(\tau,\si)=h_{[0,u)}(\tau',\si)$ implies that $\Val_u (\phi(\tau,\si))=\Val_u (\phi(\tau',\si))$. One could thus informally write $\Val_u(\phi(\tau,\si))=\Val_u(\phi(\tau_{[0,u)},\si))$ for every $\tau \in \Strat^\Aset $  and $u\in [t,T]$. In the continuous time framework, the latter property will
underpin the definition of an $(\wh{S}, \Hist )$-predictable mapping $\phi $ (see Section \ref{sec7f}).

\section{Arbitrage Bounds for a Combined Tranche}\label{sec7da}

An important concept in multi-person game contingent claims is the ability to hold a collection of tranches $\Aset\subseteq\Mset$. For example, a person holding tranches $i$ and $j$ simultaneously has possibly greater payoff earning potential than two separate people holding tranches $i$ and $j$ without collusion. This is due to the fact that the holder of tranches $i$ and $j$ may coordinate the strategies $s^i$ and $s^j$ to improve the combined payoff $V^i_{\stp (s)}(s)+V^j_{\stp (s)}(s)$. Consequently, the price of the combined tranche $\Aset = \{i,j\}$ is not necessarily equal to the sum of individual prices, even though from the practical perspective this additivity property seems to be desirable. The issue of price additivity and consistent valuation of
all tranches will be further explored in Section \ref{sec7e2}.

In this section, we fix $\Aset $ and we restrict our attention to the pricing of a predetermined combined tranche $\MGCC^\Aset_t(\Game,s)$ for some fixed $t\in[0,T]$ and a strategy profile $s\in\Strat$.
We address the following problem, in which we implicitly work on the event of $\{\stp(s)\geq t\}$ (since otherwise
the question would be meaningless anyway):
\begin{itemize}
\item Consider the game contingent claim $\MGCC(\Game)$ and suppose that a strategy profile $s_{[0,t)}$ has been played up to time $t$. At time $t$, a new holder purchases the tranche  $\Aset$ from its previous holder(s) for the price of $\pi^\Aset_t(\Game,s)$ and plans to hold it until the settlement time. What possible values can $\pi^\Aset_t(\Game,s)$ take if neither issuer's nor holder's arbitrage in $(B,S, \Aset )$ may occur?
\end{itemize}

\subsection{Arbitrage Opportunities}

We first define arbitrage opportunities from the perspective of a holder of a fixed combined tranche $\Aset $,
as well as the issuer who is assumed to hold all remaining single tranches, that is, a combined tranche $-\Aset $.
Recall the notation $s^\Aset=(s^i,\, i\in\Aset)$ and $V^\Aset_{\stp (s)}(s)=\sum_{i\in\Aset} V^i_{\stp (s)}(s)$ for any subset $\Aset\subseteq\Mset$. We stress that in the rest of this section, a date $t$, a combined tranche $\Aset $ and a strategy profile $s$ generating the history via $h_{[0,t)}(s)$ are fixed. Our goal in this section is to derive arbitrage bounds at time $t$ for a fixed combined tranche $\Aset $. Nevertheless, the specification of payoffs for tranches from the set $- \Aset $ will still matter, since decisions of the holder of $-\Aset $ may impact the payoffs of tranches from $\Aset $.

In the next definition, we fix $t \in [0,T]$ and we assume that the tranche $\Aset$ of the game contract $\MGCC_t(\Game,s)$ is traded at time $t$ at the price of $\pi_t^\Aset(\Game, s)$. To alleviate notation, we will simply write $(B,S,\Aset)$ instead of a more precise notation $(B,S,\pi^\Aset_t(\Game,s))$ if no confusion may arise. 

\begin{definition} \label{defrm03} {\rm
A \emph{holder's arbitrage in $(B,S,\Aset )$} is a pair $(\phi^{\tau },\tau)\in \Phi_t (\Hist^\tau)\times\Strat^\Aset_t(s)$ for which there exists an event $A \in\Filt_t$ with a positive probability such that on $A$:
\[
\Val_t(\phi^{\tau } (s) ) < - \pi_t^\Aset(\Game, s) ,\quad
\Val_{\stp{(s)}}(\phi^{\tau } (s) ) \geq - V^\Aset_{\stp(s)}(\tau,\si),\quad\forall\,\si\in\Strat^{-\Aset}_t(s).
\]
An \emph{issuer's arbitrage in $(B,S,\Aset )$} is a pair $(\phi^{\si } ,\si)\in \Phi_t (\Hist^\si)\times\Strat^{-\Aset}_t(s)$ for which there exists an event $A \in\Filt_t$ with a positive probability such that on $A$:
\[
\Val_t(\phi^{\si } (s)) <  \pi_t^\Aset(\Game, s) ,\quad
\Val_{\stp{(s)}}(\phi^{\si } (s)) \geq V^\Aset_{\stp(s)}(\tau,\si),\quad \forall\,\tau\in\Strat^\Aset_t(s).
\]
We say that there is \emph{no arbitrage in $(B,S, \Aset )$} if neither the issuer's nor the holder's arbitrage exists in $(B,S,\Aset)$.}
\end{definition}

Note that Definition \ref{defrm03} hinges on an implicit assumption that tranches are traded at time $t$ before
actions at time $t$ are chosen by all parties. In other words, only the history $h_{[0,t)}$ and the $\sigma$-field
$\Filt_t$ are observed by all parties before trading occurs at time $t$.

A holder's arbitrage refers to a guaranteed profit for the buyer of the tranche $\Aset$ on $\cll t,\stp(s) \crr $. The price at
time $t$ should be high enough, so that a prospective holder should not be able to generate profits without risk by
combining the buy-and-hold strategy with some clever actions and a dynamic portfolio composed of traded primary assets. Similarly, an issuer's arbitrage refers to a guaranteed profit for issuer of $\MGCC(\Game)$ who sold the tranche $\Aset$ at time $t$ and also holds the tranche $-\Aset$ until time $\stp(s)$. From this perspective, the price should be low enough, so that it prevents the issuer from making profit without risk even when he can make decisions regarding the tranche $-\Aset $. When the issuer does not hold the tranche $-\Aset$, we may interpret an issuer's arbitrage as a guaranteed loss for the holder of the tranche $\Aset$.
It is rather clear that, under Definition \ref{defrm03}, the valuation problem for the tranche $\Aset$ can be formally reduced
to a study of the two-person zero-sum game between the holder of the tranche $\Aset $ and the issuer, who also holds $-\Aset $.

Arbitrages introduced in Definition \ref{defrm03} hinge on the enter-and-hold strategy, where the tranches are kept from time $t$ until the settlement time. Of course, alternative definitions of arbitrage opportunities are also possible. In Section \ref{sec7e2}, we will complement Definition \ref{defrm03} with additional conditions that will ensure the existence and uniqueness of financially meaningful unique prices for all tranches.


\subsection{Super-Hedging Strategies}

For a European contingent claim, the terminal payoff can be frequently \emph{replicated} by means of an admissible portfolio. In the case of a game contingent claim, since it is not possible to anticipate the action of other parties, replication is typically not possible. Instead, the parties involved may attempt to \emph{super-hedge} their positions. The definition of super-hedging trading strategies is very similar to that for the two-person game option. As before, we consider the time $t$ game contingent claim $\MGCC_t(\Game,s)$. We focus on the holder of the tranche  $\Aset$ and the issuer, who also holds the tranche  $-\Aset$.
Recall that we denote $s = (\tau,\si)$ and we place ourselves at time $t \in [0,T]$.

\begin{definition} \label{defrm70} {\rm
A \emph{holder's super-hedging strategy} is a pair $(\phi^\tau,\tau)\in \Phi_t(\Hist^\tau)\times\Strat^{\Aset}_t(s)$ such that
\[
\Val_{\stp{(s)}}(\phi^\tau) \geq - V^\Aset_{\stp(s)}(\tau,\si),\quad\forall\,\si\in\Strat^{-\Aset}_t(s).
\]
An \emph{issuer's super-hedging strategy} is a pair $(\phi^\si,\si)\in \Phi_t(\Hist^\si)\times\Strat^{-\Aset}_t(s)$ such that
\[
\Val_{\stp{(s)}}(\phi^\si) \geq V^\Aset_{\stp(s)}(\tau,\si),\quad \forall\,\tau\in\Strat^\Aset_t(s).
\]}
\end{definition}

As expected, the super-hedging strategies are closely related to arbitrage bounds for prices of tranches of game contingent claims.

\begin{proposition}\label{proprm71}
Suppose that the tranche  $\Aset$ of the game contract $\MGCC_t(\Game,s)$ is traded at time $t$ at a price of $\pi_t^\Aset(\Game, s)$. If no issuer's arbitrage in $(B,S,\Aset)$ exists, then for any issuer's super-hedging strategy $(\phi^\si,\si)\in \Phi_t(\Hist^\si)\times\Strat^{-\Aset}_t(s)$ we have that
\begin{gather}\label{eqrm711}
\Val_t(\phi^\si) \geq \pi_t^\Aset(\Game, s).
\end{gather}
If no holder's arbitrage in $(B,S,\Aset)$ exists, then for any holder's super-hedging strategy $(\phi^\tau,\tau)\in \Phi_t(\Hist^\tau)\times\Strat^{\Aset}_t(s)$  we have that
\begin{gather}\label{eqrm712}
\Val_t(\phi^\tau) \geq - \pi_t^\Aset(\Game, s).
\end{gather}
\end{proposition}

\begin{proof}
Let us show that if the upper bound \eqref{eqrm711} is violated, then there exists an issuer's arbitrage. A similar argument can be used to check that the violation of the lower bound \eqref{eqrm712} leads to a holder's arbitrage. For the sake of contradiction, let us assume the upper bound of \eqref{eqrm711} is not satisfied. On the event $\big\{\pi^\Aset_t(\Game,s) > \Val_t(\phi^\si)\big\}$, by the definition of an issuer's super-hedging strategy, we obtain
\[
\Val_t(\phi^\si) < \pi^\Aset_t(\Game,s), \quad \Val_{\stp{(s)}}(\phi^\si) \geq V^\Aset_{\stp(s)}(\tau,\si),\quad \forall\,\tau\in\Strat^\Aset_t(s).
\]
Hence $(\phi^\si,\si)$ is an issuer's arbitrage in $(B,S,\Aset)$, which contradicts our assumption.
\end{proof}

\subsection{$\Hist^{\si }$-Predictable Snell Envelopes} \label{sec3.3x}

The notion of the Snell envelope, first introduced in \cite{Snell}, is an important tool for the valuation of American and game options. For a given stochastic process $Y$, the \emph{Snell envelope} $U$ is the smallest RCLL supermartingale dominating $Y$. At time $t$, the Snell envelope $U_t$ equals $\esssup_{\tau\in\STOP_{[t,T]}}\E(Y_\tau\,|\,\Filt_t)$, so that it represents the essential supremum of $\Filt_t$-conditional expectations of $Y_\tau$ over all choices of a stopping time $\tau \in \STOP_{[t,T]}$.

A similar notion will be developed here in the context of a multi-person game claim $\MGCC(\Game)$. It will hinge on maximizing an expected value over the set of all possible strategies, rather than the set of all stopping times. Recall that the subset $\Aset \in \Mset $ is now fixed and we deal with a combined tranche, specifically, a collection of single tranches $i$ for all $i \in \Aset $. Unless otherwise stated, we fix a strategy profile $s = (\tau,\si) \in\Strat$ where $\tau \in\Strat^{\Aset}$ and $\si\in\Strat^{-\Aset}$.

Recall that we work throughout under the unique martingale measure $\Qrob$ for a discrete-time (or a continuous-time) market model, which is assumed to be complete and arbitrage-free in the usual sense. Let $\widehat V_{\stp (s)}(s) = B^{-1}_{\stp(s)} V_{\stp (s)} (s)$ be the discounted payoff and let $\widehat V^\Aset_{\stp (s)} (s) = \sum_{i \in \Aset } \widehat V^i_{\stp (s)}(s)$.
In the remainder of this work, we work under the following standing assumption.

\begin{assumption} \label{assrm09} {\rm
The discounted payoff satisfies the following integrability condition with respect to $\Qrob$, for every $\Aset\subseteq\Mset$,}
\[
\EMM \Big( \esssup_{s\in\Strat} \big| \widehat V^\Aset_{\stp (s)} (s)\big| \Big) < \infty.
\]
\end{assumption}

\subsubsection{$\Hist^{\si }$-Predictable Snell Envelope of the First Kind}

We first define the $\Hist^{\si }$-predictable Snell envelope of the first kind, which is directly based on the observations of decisions of the holder of the tranche $\Aset $ strictly before time $t$ and the market data up to time $t$. Note that the holder of $- \Aset$ has a priori the full knowledge of her fixed strategy $\si $ on $[0,T]$, but she only observes the outcome $h_{[0,t)}(\tau,\si)$ at time $t$.

\begin{definition}\label{defrm44}{\rm
For a combined tranche $\MGCC^\Aset(\Game)$ and a fixed $\si  \in \Strat^{-\Aset }$, the
\emph{$\Hist^{\si }$-predictable Snell envelope of the first kind} is the mapping $U^\si \map{\Hist^\si \times[0,T]\times\Omega}{\R}$ given by}
\begin{gather}\label{eqrm45}
U^\si_t (\tau ) = U^\si_t(h(\tau,\si)) := \esssup_{\tau'\in\Strat^\Aset_t(\tau,\si)} \EMM \Big(\widehat
V^{\Aset}_{\stp(\tau',\si )}(\tau',\si) \,\big|\,\Filt_t\Big).
\end{gather}
\end{definition}

\begin{remark} {\rm
Let us note that the definition of $U^\si_t( \tau )$ is not restricted to the event $\{\stp(\tau,\si)\geq t\}$. Indeed, on the event $\{\stp(\tau,\si)<t\}$, by convention, we may set}
\[
U^\si_t( \tau) : =\widehat V^{\Aset}_{\stp(s)}(\tau,\si) \leq U^\si_{\stp (s)}(\tau ).
\]
\end{remark}

Recall that the class $\Strat^\Aset_t(\tau,\si)$ is defined by
\begin{gather}\label{eqrm46}
\Strat^\Aset_t(\tau,\si) = \big\{ \tau' \in\Strat^\Aset : h_{[0,t)}(\tau',\si) =h_{[0,t)}(\tau,\si) \big\}.
\end{gather}
Since $\si$ and $\tau $ are here fixed, the class $\Strat^\Aset_t(\tau,\si)$ only depends on $h_{[0,t)}(\tau,\si)$. Therefore, the right-hand side in \eqref{eqrm45} is indeed a well-defined mapping of the outcome $h(\tau,\si)\in\Hist^\si$, rather than the strategy $(\tau,\si)$. Despite this, for the sake of brevity, we prefer to write $U^\si_t(\tau)$, rather than $U^\si_t(h(\tau,\si))$.

 The superscript $\si $ in $U^\si(\tau)$ is used to emphasise that the strategies $\si $ and $\tau $
play different roles in Definition \ref{defrm44}. Specifically, at any date $t$ we assume to know an issuer's strategy $\si $ over $[0,T]$, but we only assume that a holder's strategy $\tau $ is observed on $[0,t).$
This feature is formalised in the next result, which shows that, intuitively, $U^\si_{t}(\tau)=U^\si_{t}(\tau_{[0,t)})$ for all $t \in [0,T]$, where $\tau_{[0,t)}$ is the restriction of the holder's strategy $\tau $ to the interval $[0,t)$.

\begin{lemma} \label{lmcc7a}
The mapping $U^\si $ is $\Hist^\si$-predictable.
\end{lemma}

\begin{proof}
It is obvious from \eqref{eqrm45} that $\tau \in\Strat^\Aset$ the random variable $U^\si_t(\tau)$ is $\Filt_t$-measurable.
For the $\Hist^\si$-predictability of $U^\si $, we need to show that, for all $\tau,\tau' \in\Strat^\Aset$, the equality
$U^\si (\tau ) =U^\si (\tau')$ holds on $[\![ 0, \rho ]\!]$ where (see \eqref{veqsg40})
\begin{gather}  \label{vueqsg40}
\rho := \inf \big\{ t\in[0,T]: h_t(\tau,\si) \neq h_t (\tau',\si) \big\} \wedge T.
\end{gather}
Recall that, by Assumption \ref{aspmm09}, the payoff mapping $V^\Aset $ (and thus also the discounted
payoff $\wh{V}^\Aset $) is $\Hist$-adapted (though not necessarily $\Hist^\si$-predictable). For any $t \in [0,T]$,
the event $\{ t \leq \rho \}$ belongs to $\Filt_t$ and the sets
\[
\big\{\widehat V^{\Aset}_{\stp(\wt \tau,\si)}(\wt \tau,\si) \I_{\{ t \leq \rho \}} : \wt \tau\in \Strat^\Aset_t(\tau,\si)\big\} \quad\text{and}\quad \big\{\widehat V^{\Aset}_{\stp(\wh \tau,\si )}(\wh \tau,\si) \I_{\{ t \leq \rho \}} : \wh \tau\in \Strat^\Aset_t(\tau',\si)\big\}
\]
are identical, due to the definition of $\Strat^\Aset_t(\tau,\si)$ and $\Strat^\Aset_t(\tau',\si)$ (see \eqref{eqrm46}). Consequently, the equality $U^\si_t (\tau ) =U^\si_t (\tau')$ holds on  $\{ t \leq \rho \}$ for every $t$.
\end{proof}

The $\Hist^{\si }$-predictable  Snell envelope $U^\si$ has the following interpretation: suppose that the holder of $-\Aset$ plays the strategy $\si\in\Strat^{-\Aset}$ on $[0,T]$ while the holder of $\Aset$ plays the strategy $\tau\in\Strat^\Aset$ on $[0,t)$. Then $U^\si_t(\tau)$ is the maximal expected payoff, as seen at time $t$, which can be achieved by the holder of $\Aset$ by varying his strategy on $[t,T]$. It is thus natural to conjecture that the process $U^\si_t(\tau)$ is a $(\Qrob , \FF)$-supermartingale. The proof of the main result in this subsection, in which this conjecture is validated  (Proposition \ref{proprm50}), requires several lemmas.

\begin{lemma} \label{proprm31}
Fix $t\in[0,T]$, $\Aset \subset \Mset$ and $s\in\Strat$. Then the following set ${\cal D}_t$ has the lattice property
\begin{align*}
{\cal D}_t : = \Big\{ \EMM \Big(\widehat V^{\Aset}_{\stp(\tau',\si)}(\tau',\si) \,\big|\,\Filt_t\Big):
\tau' \in\Strat^\Aset_t(\tau,\si)\Big\}.
\end{align*}
\end{lemma}

\begin{proof}
Let $I$ be an arbitrary set and let ${\cal Y}:= (Y_i)_{i\in I}$ be a family of random variables with $\EMM (\esssup_{i\in I} |Y_i|)<\infty$. Then ${\cal Y}$ is said to have the \emph{lattice property} if, for all $i,j\in I$, there exist $k, l\in I$ such that
\[
Y_k\geq Y_i \vee Y_j,\quad Y_l\leq Y_i \wedge Y_j .
\]
In our case, the set ${\cal D}_t$ may be represented as $\{ Y_{\tau'} : \tau'\in\Strat^\Aset_t(s)\}$ where each $Y_{\tau'}$ is $\Filt_t$-measurable and, in view of Assumption \ref{assrm09},
\[
\EMM(\esssup_{\tau'\in\Strat^\Aset_t(s)} |Y_{\tau'}|)<\infty .
\]
We first show that it has the following property: for any event $A\in\Filt_t$ and arbitrary $\tau',\tau''\in\Strat^\Aset_t(s)$
\begin{gather}\label{eqrm310}
\tau' \I_A = \tau'' \I_A \ \implies\ Y_{\tau'} \I_A = Y_{\tau''} \I_A .
\end{gather}
 In view of the $\Hist$-adaptedness of the settlement time $\stp $ and the discounted payoff $\wh{V}^{\Aset }_{\stp (s)}(\tau,\si)$, the equality ${\tau'} \I_A = {\tau''} \I_A$ implies that
\[
\wh{V}^{\Aset }_{\stp(\tau',\si)}(\tau',\si) \I_A = \wh{V}^{\Aset }_{\stp(\tau'',\si)}(\tau'',\si)\I_A .
\]
Since $\I_A$ is $\Filt_t$-measurable, we obtain $Y_{\tau'} \I_A = Y_{\tau''} \I_A$, thus establishing \eqref{eqrm310}.
Let us now check the lattice property of ${\cal D}_t$. For arbitrary $\tau', \tau'' \in \Strat^\Aset_t(s)$, we define the strategies $\tau_1,\tau_2$ by setting
\begin{align*}
&\tau_1 = \I_{\{Y_{\tau'} \geq Y_{\tau''}\}} \tau' + \I_{\{Y_{\tau'} < Y_{\tau''}\}} \tau'' , \\
&\tau_2 = \I_{\{Y_{\tau'} \geq Y_{\tau''}\}} \tau'' + \I_{\{Y_{\tau'} < Y_{\tau''}\}} \tau'.
\end{align*}
Since the event $\{Y_{\tau'} \geq Y_{\tau'}\}$ is $\Filt_t$-measurable, it is clear that $\tau_1, \tau_2$ are $\FF$-adapted, $\Hist$-predictable mappings satisfying
\[
h_{[0,t)}(\tau_1,\si)=h_{[0,t)}(\tau_2,\si)=h_{[0,t)}(\tau',\si)=h_{[0,t)}(\tau'',\si)=h_{[0,u)}(s), \quad\forall\,\si\in\Strat^{-\Aset}_t(s).
\]
Hence $\tau_1$ and $\tau_2$ belong to $\Strat^\Aset_t(s)$. Using condition \eqref{eqrm310}, we also obtain $Y_{\tau_1}=Y_{\tau'}\vee Y_{\tau''}$ and $ Y_{\tau_2}=Y_{\tau'}\wedge Y_{\tau''}$. This establishes the lattice property of the set ${\cal D}_t$.
\end{proof}

The following result is a minor extension of Theorem A.32 in \cite{Follmer} (for the proof, see Lemma 7.18 in \cite{GuoPhD}).

\begin{lemma}\label{lemrm30}
Let $(Y_i)_{i\in I}$ be a family of $\Filt_u$-measurable random variables with the lattice property and such that $\E(\esssup_{\, i\in I} |Y_i|)<\infty$.  Then the following statements are valid:
\\(i) there exist two sequences $(i_n)_{n\in\N}, (j_n)_{n\in\N}$ of indices such that the sequences $(Y_{i_n})_{n\in\N}$ and $(-Y_{j_n})_{n\in\N}$  are almost surely non-decreasing and
\[
\esssup_{i\in I} Y_i = \sup_{n\in\N} Y_{i_n} = \lim_{n\to\infty} Y_{i_n},\quad
\essinf_{j\in I} Y_j = \inf_{n\in\N} Y_{j_n} = \lim_{n\to\infty} Y_{j_n},
\]
(ii) for any $t \leq u$, we have
\[
\EMM \big(\esssup_{i\in I} Y_i \,\big|\, \Filt_t \big)=\esssup_{i\in I} \E\big( Y_i \,|\, \Filt_t\big),\quad
\EMM \big(\essinf_{j\in I} Y_j \,\big|\, \Filt_t \big)=\essinf_{j\in I} \E\big( Y_j \,|\, \Filt_t\big).
\]
\end{lemma}


%


\begin{lemma} \label{lemm50x}
(i) The following equality holds, for any $t\leq u$,
\begin{align}\label{eqrm505}
\EMM \big(U^\si_u(\tau)\,\big|\,\Filt_t\big) =\esssup_{\tau'\in\Strat^\Aset_u(\tau,\si)} \EMM \Big( \widehat V^{\Aset}_{\stp (\tau',\si)} (\tau',\si)\,\big|\,\Filt_t\Big).
\end{align}
(ii) For any $\tau \in\Strat^{\Aset} $, the process $U^\si (\tau)$ is uniformly integrable under $\Q$.
\end{lemma}

\begin{proof}
(i) Let us consider $t \leq u$. From Lemma \ref{proprm31}, we deduce that, for fixed $\si \in\Strat^{-\Aset}$
and $\tau \in\Strat^{\Aset}$, the set
\begin{align}\label{caldu}
{\cal D}_u = \Big \{\EMM \Big( \widehat V^{\Aset}_{\stp (\tau',\si)}(\tau',\si) \,\big|\,\Filt_u\Big): \tau'\in\Strat^{\Aset}_u (\tau,\si) \Big\}
\end{align}
has the lattice property. Using part (ii) in Lemma \ref{lemrm30} and the tower property of conditioning, we thus obtain
\begin{align*}
\EMM \big(U^\si_u(\tau)\,\big|\,\Filt_t\big)
 = \EMM \Big(\esssup_{\tau'\in\Strat^\Aset_u(\tau,\si)}
\EMM \Big( \widehat V^{\Aset}_{\stp (\tau',\si)} (\tau',\si) \,\big|\,\Filt_u\Big) \,\Big|\,\Filt_t \Big)
=\esssup_{\tau'\in\Strat^\Aset_u(\tau,\si)} \EMM \Big( \widehat V^{\Aset}_{\stp (\tau',\si)} (\tau',\si) \,\big|\, \Filt_t\Big)
\end{align*}
so that \eqref{eqrm505} is valid.

\noindent (ii)  In view of Assumption \ref{assrm09}, the process $M$ given by $M_t = \EMM (M_T \,|\, \Filt_t)$, where
\[
M_T := \esssup_{s\in\Strat} \big|\widehat{V}^\Aset_{\stp (s)}(s)\big|,
\]
is a non-negative and uniformly integrable $(\Qrob , \mathbb{F})$-martingale. Using the conditional Fatou lemma, we obtain, for any $t\in [0,T]$,
\begin{gather*}
\big|U^\si_t(\tau)\big| =  \Big|\esssup_{\tau'\in\Strat^\Aset_t(\tau,\si)} \EMM \Big( \widehat V^{\Aset}_{\stp (\tau',\si)} (\tau',\si) \,\big|\,\Filt_t\Big)\Big| \leq \esssup_{\tau'\in\Strat^\Aset_t(\tau,\si)} \EMM \Big( | \widehat V^{\Aset}_{\stp (\tau',\si)} (\tau',\si)| \,\big|\,\Filt_t\Big)  \\
\leq \EMM \Big( \esssup_{\tau'\in\Strat^\Aset_t(\tau,\si)} \big|\widehat V^{\Aset}_{\stp (\tau',\si)} (\tau',\si)\big| \,\Big|\,\Filt_t\Big) \leq \EMM \Big(\esssup_{s\in\Strat} \big|\widehat V^\Aset_{\stp (s)} (s)\big|\,\Big|\,\Filt_t\Big) = M_t.
\end{gather*}
Hence the process $U^\si (\tau)$ is uniformly integrable under $\Q$.
\end{proof}

\begin{proposition} \label{proprm50}
(i) For any $\tau \in\Strat^{\Aset} $, the process $U^\si(\tau)$ is a $(\Qrob , \FF)$-supermartingale. \\
(ii) In the continuous-time case, if the process $U^\si (\tau)$ has an RCLL modification, then it
is a $(\Qrob , \FF)$-supermartingale of class (D).
\end{proposition}

\begin{proof} (i) It is clear that $U^\si(\tau)$ is an $\FF$-adapted process. The supermartingale property of $U^\si(\tau)$ follows immediately from part (i) in Lemma \ref{lemm50x} and the obvious inclusion $\Strat^\Aset_u(\tau,\si) \subseteq \Strat^\Aset_t(\tau,\si)$ for all $t \leq u$, since
\begin{align*}
\EMM \big(U^\si_u(\tau)\,\big|\,\Filt_t\big) & = \esssup_{\tau'\in\Strat^\Aset_u(\tau,\si)} \EMM \Big( \widehat V^{\Aset}_{\stp(\tau',\si)}(\tau',\si)\,\big|\,\Filt_t\Big) \\ &\leq \esssup_{\tau'\in\Strat^\Aset_t(\tau,\si)}
\EMM \Big( \widehat V^{\Aset}_{\stp(\tau',\si)}(\tau',\si) \,\big|\,\Filt_t\Big) = U^\si_t(\tau).
\end{align*}
\noindent (ii) Recall that a process $X$ is said to be of \emph{class (D)} if it is RCLL and the family $\{X_\rho : \rho\in\STOP_{[0,T]}\}$ of random variables is uniformly integrable. We assumed that $U^\si (\tau)$
has an RCLL modification. Let $M$ be an RCLL version of the martingale $M$. Then $|U^\si_\rho(\tau)\big|\leq M_\rho$ for any stopping time $\rho\in\STOP_{[0,T]}$ where the family $\{ M_\rho :\rho\in\STOP_{[0,T]}\}$ is uniformly integrable.
Hence the family $\{U^\si_\rho(\tau):\rho\in\STOP_{[0,T]}\}$ is uniformly integrable and thus $U^\si(\tau)$ is a supermartingale of class~(D).
\end{proof}

\subsubsection{$\Hist^{\si }$-Adapted Snell Envelope}

Definition \ref{defrm44x} of the $\Hist^{\si }$-adapted Snell envelope at time $t$ hinges on the assumption
that the decisions of all parties at time $t$ are already known. It other words, it is now assumed
that, at time $t$, the holder of $-\Aset $ observes the market data up to time $t$, as well as
the outcome $h_{[0,t]}(\tau,\si)$. It is thus rather clear that the $\Hist^{\si }$-adapted Snell envelope
is not directly suitable for super-hedging purposes. However, it will serve as a crucial tool in the introduction
of the $\Hist^{\si }$-predictable Snell envelope of the second kind in the continuous-time setup
(see Section \ref{seci7f}).

\begin{definition}\label{defrm44x}{\rm
For a combined tranche $\MGCC^\Aset(\Game)$ and a fixed $\si \in \Strat^{-\Aset}$, we define the \emph{$\Hist^{\si }$-adapted Snell envelope} $\wt{U}^\si \map{\Hist^\si \times[0,T]\times\Omega}{\R}$ by setting
\begin{gather}\label{eoqrm45}
\wt{U}^\si (\tau) = \wt{U}^\si_t( h(\tau,\si) ) := \esssup_{\tau'\in \wt{\Strat}^\Aset_t(\tau,\si)} \EMM \Big(\widehat V^{\Aset}_{\stp(\tau',\si)}(\tau',\si) \,\big|\,\Filt_t\Big)
\end{gather}
where the subset $\wt{\Strat}^\Aset_t(\tau,\si) \subseteq\Strat^\Aset$ is given by}
\begin{equation} \label{defstrx}
\wt{\Strat}^\Aset_t(\tau,\si) := \big\{ \tau' \in\Strat^\Aset : h_{[0,t]}(\tau',\si) =h_{[0,t]}(\tau,\si) \big\} .
\end{equation}
\end{definition}

\begin{lemma} \label{lmcc7b}
(i) The mapping $\wt U^\si $ is $\Hist^\si$-adapted. \\
(ii) Fix $t\in[0,T]$, $\Aset \subset \Mset$ and $s\in\Strat$. Then the following set ${\cal \wt D}_t$ has the lattice property
\begin{align*}
{\cal \wt D}_t : = \Big\{ \EMM \Big(\widehat V^{\Aset}_{\stp(\tau',\si )}(\tau',\si ) \,\big|\,\Filt_t\Big): \tau' \in\wt \Strat^\Aset_t(\tau,\si)\Big\}.
\end{align*}
\end{lemma}

\begin{proof}
The proof of part (i)  (part (ii), resp.) is analogous to the proof  of Lemma \ref{lmcc7a}
(Lemma \ref{proprm31}, resp.) and thus it is omitted.
\end{proof}


\begin{proposition} \label{prprm50}
(i) The following equality holds, for any $t\leq u$,
\begin{align*} 
\EMM \big(\wt U^\si_u(\tau)\,\big|\,\Filt_t\big) =\esssup_{\tau'\in\wt \Strat^\Aset_u(\tau,\si)} \EMM \Big( \widehat V^{\Aset}_{\stp (\tau',\si)} (\tau',\si)\,\big|\,\Filt_t\Big).
\end{align*}
(ii) The process $\wt U^\si (\tau)$ is uniformly integrable under $\Q$. In particular,
\[
\EMM \Big(\esssup_{t\in[0,T]}\big| \wt U^\si_t(\tau)\big|\Big) < \infty.
\]
(iii) For any $\tau \in\Strat^{\Aset} $, the process $\wt U^\si(\tau)$ is a $(\Qrob , \FF)$-supermartingale.
\end{proposition}

\begin{proof}
In view of Lemma \ref{lmcc7b}, all statements  can be established in a similar way as analogous statements
for the mapping $U^\si(\tau)$ (see the proofs of Lemma \ref{lemm50x} and Proposition \ref{proprm50}).
\end{proof}

\subsection{Discrete-Time Case}\label{sec7e}

We now assume a discrete-time arbitrage-free and complete market model $(B,S)$ with a finite underlying probability
 space $\Omega$. Therefore, for any European contingent claim with the payoff $X$ at its maturity date $T$, there exists a self-financing trading strategy $\phi$ satisfying,
for all $t \in [0,T]$,
\[
Z_t (\phi ) := \sum_{i=0}^d \phi^i_{t} S^i_t = \EMM (X_T \,|\, \Filt_t ).
\]
Note that the portfolio $\phi_t$ is established at time $t-1$, so that it is $\Filt_{t-1}$-measurable, meaning that
the process $\phi $ is $\FF$-predictable, and $\phi_{t}$ may depend on the outcome $h_{[0,t-1]}$.  At time $t$, the portfolio is revised from $\phi_t$ to $\phi_{t+1}$ such that
\[
\sum_{i=0}^d \phi^i_{t} S^i_t = \sum_{i=0}^d \phi^i_{t+1} S^i_t ,
\]
$\phi_{t+1}$ is $\Filt_{t}$-measurable and it may depend on $h_{[0,t]}$.  Since the model is arbitrage-free, all trading strategies that replicate $X$ have the same wealth process. The uniqueness of a replicating strategy is not guaranteed, however, since the complete market $(B,S)$ may still have redundancies.

We return to the question posed in the beginning of the section, that is, for a given tranche $\Aset $ we search for the possible prices of $\pi^\Aset_t(\Game,s)$ for some fixed $t \in [0,T]$ and $s \in \Strat $ that avoid arbitrage in $(B,S,\Aset)$. We will first focus on the discrete-time case, before proceeding to the continuous-time one. Recall the convention $[0,t)=[0,t-1]:= \{0, \dots , t-1\}$ for every $t=1,2, \dots ,T$. We denote by $\wh{Z}(\phi )$ the \emph{discounted wealth} given as $\wh{Z}(\phi )=B^{-1} Z(\phi ).$ We will need the following lemma.

\begin{lemma}\label{propro71}
Suppose that a mapping $\wh{M}\map{\Hist\times[0,T]\times\Omega}{\R}$ satisfies the following properties: \hfill \break
(i) for any $s\in\Strat$, the process $\wh{M}(s)=\wh{M}(h(s))$ is a $(\Qrob , \mathbb{F})$-martingale,
\\ (ii) $\wh{M}$ is $\Hist$-predictable. \\
Then there exists an  $\Hist$-predictable mapping  $\phi$  such that, for every $s \in \Strat $, the discounted wealth of a trading strategy $\phi (s)$ satisfies $\widehat Z(\phi (s))=\wh{M}(s)$.
\end{lemma}

\begin{proof}
For any $s \in \Strat$, we simply define $\phi(s)$ as a replicating (hence necessarily admissible) trading strategy for the claim $B_T \wh{M}_T(s)$. Hence the existence of $\phi (s)$ follows immediately from the postulated completeness of the discrete-time market model $(B,S)$ and $\Qrob$-integrability of $\wh{M}_T(s)$. It is also not hard to show that an $\Hist$-predictable version of the mapping $\phi $ can be selected (for details, see Section 7.2.2 in Guo \cite{GuoPhD}).
\end{proof}

\begin{remark} \label{remtt} {\rm
Note that the statement of the lemma will still hold if we replace $\Hist$ and $[0,T]$ by $\Hist^\si$ and $[t,T]$, respectively.}
\end{remark}

The next result shows that, for any fixed $(\tau,\si)\in\Strat$ and $t \in [0,T]$, there exists an issuer's super-hedging strategy $(\phi^\si,\si)$  on $[t,T]$ whose discounted wealth at time $t$  coincides with the value $U^\si_t(\tau)$ of the $\Hist^{\si }$-predictable  Snell envelope of the first kind. This property is a essential tool in establishing the upper bound in Theorem \ref{thmro10}.

\begin{proposition} \label{propro01}
 Take an arbitrary $\si\in\Strat^{-\Aset}$ and $t\in[0,T]$. If a strategy $\tau\in\Strat^{\Aset}$ is played by the holder on $[0,t-1]$, then there exists a trading strategy $\phi^\si (\tau ) \in\Phi_t(\Hist^\si)$ such that $(\phi^\si (\tau ),\si)$ is an issuer's super-hedging strategy on $[t,T]$ and its discounted time $t$ wealth is equal to the value  $U^\si_t(\tau)$ of the $\Hist^{\si }$-predictable  Snell envelope of the first kind. More explicitly, the discounted wealth process $\widehat \Val_u (\phi^\si(\tau)) := B^{-1}_u\Val_u(\phi^\si(\tau)),\, u \in [t,T]$, satisfies, on the event $\{\stp(\tau,\si)\geq t\}$,
\begin{gather}\label{eqro011}
\widehat \Val_t(\phi^\si(\tau))=U^\si_t(\tau)
\end{gather}
and
\begin{gather}\label{eqro012}
\widehat \Val_{\stp(\tau,\si)}(\phi^\si(\tau))\geq U^\si_{\stp(\tau,\si)}(\tau) \geq \widehat V^{\Aset}_{\stp(\tau,\si)}(\tau,\si).
\end{gather}
\end{proposition}

\begin{proof}
We fix $\si\in\Strat^{-\Aset}$ and $\tau\in\Strat^{\Aset}$. Recall that the Snell envelope $U^\si $
is $\Hist^{\si }$-predictable and,  by Proposition \ref{proprm50}, for any $\tau \in \Strat^{\Aset }$, the process $U^\si(\tau)$  is a $(\Qrob , \mathbb{F})$-supermartingale. By applying the Doob decomposition theorem, we obtain the unique decomposition $U^\si(\tau) = M^\si(\tau) - A^\si(\tau)$ where  $M^\si(\tau)$ is a $(\Qrob , \mathbb{F})$-martingale and $A^\si(\tau)$  is an increasing, integrable, ${\mathbb F}$-predictable process such that $A^\si_0(\tau)=0$.

Let us take arbitrary $\tau,\tau'\in \Strat^\Aset$ and let us define the $\FF$-stopping time
\begin{align} \label{tautau}
\rho = \rho (\tau , \tau ') := \inf \big\{ u\in[0,T]: h_u(\tau,\si) \neq h_u(\tau',\si) \big\} \wedge T.
\end{align}
Since the mapping $U^\si $ is $\Hist^{\si}$-predictable the equality $U^\si(\tau)=U^\si(\tau')$ holds
on $\cll 0 , \rho \crr $, meaning that $U^\si_{\cdot\wedge\rho}(\tau)=U^\si_{\cdot\wedge\rho}(\tau')$.
By the uniqueness of the Doob decomposition, we thus obtain
\begin{gather} \label{ju7}
M^\si_{\cdot\wedge\rho}(\tau) = M^\si_{\cdot\wedge\rho}(\tau'),
\quad A^\si_{\cdot\wedge\rho}(\tau) = A^\si_{\cdot\wedge\rho}(\tau'),
\end{gather}
For a fixed $t$, we define the process $\widehat{M}^\si_u (\tau) := M^\si_u (\tau)-A^\si_t(\tau), \, u \in [t,T]$, which is a  $(\mathbb{Q},\FF)$-martingale. Then we deduce from \eqref{ju7} that $\wh{M}^\si(\tau)=\wh{M}^\si(\tau')$ on $\cll t , \rho \crr $.
We conclude that the mapping $\widehat{M}^\si \map{\Hist^\si \times[t,T]\times\Omega}{\R}$ is $\Hist^\si$-predictable and thus it satisfies the assumptions of Lemma \ref{propro71} on $[t,T]$ (see also Remark \ref{remtt}).

Let $\phi^\si \in\Phi_t(\Hist^\si)$ be the mapping given by Lemma \ref{propro71} and Remark \ref{remtt}, for the mapping $\widehat{M}^\si $. To establish \eqref{eqro011} and \eqref{eqro012}, we observe that, for all $u \in [t,T]$,
\[
\widehat\Val_u(\phi^\si(\tau)) =\widehat{M}^\si_u (\tau)  = M^\si_u(\tau)-A^\si_t(\tau) = U^\si_u(\tau)+A^\si_u(\tau)-A^\si_t(\tau)\geq U^\si_u(\tau),
\]
where the inequality becomes equality when $u=t$. This shows that all inequalities in \eqref{eqro012} are valid.
\end{proof}

\subsubsection{No-Arbitrage Bounds in Discrete Time}

Proposition \ref{propro01} has the following interpretation. Suppose that the issuer also holds the tranche $-\Aset$
and chooses the strategy $\si $. Assume also we are at time $t$ and the holder of the tranche $\Aset $ played a
strategy $\tau $ on $[0,t-1]$. Then by implementing the $\Hist^\si$-adapted trading strategy $\phi^\si (\tau )$, the issuer is super-hedging his position for all $u \in [t,T]$, so that the wealth $\Val_u(\phi^\si(\tau))$ of his portfolio will always cover the required payoff for the holder of the tranche $\Aset$, no matter how the holder of $\Aset $ will decide to play on $[t,T]$. By construction, $\phi^\si (\tau )$ is also the cheapest (as of time $t$) of such super-hedging trading strategies for a predetermined $\si$, and the discounted wealth at time $t$ of the issuer's portfolio is equal to the value $U^\si_t(\tau)$ of the $\Hist^{\si }$-predictable Snell envelope of the first kind.

By Proposition \ref{proprm71}, an issuer's arbitrage is precluded when the price of the tranche  $\Aset$ is bounded above by the value of any issuer's super hedging strategy. This implies that $U^\si_t(\tau)$ is an upper bound of the discounted arbitrage price at time $t$. In Theorem \ref{thmro10}, we show that the best upper bound is obtained by taking the infimum of $U^\si_t(\tau)$ over all possible choices of $\si\in\Strat_t^{-\Aset}$.

\begin{theorem}\label{thmro10}
Let us fix $\Aset\subseteq\Mset$ and let us consider the tranche $\MGCC^\Aset_t(\Game,s)$. There is no arbitrage in $(B,S,\Aset)$
at time $t$ if and only if the discounted price $\widehat \pi^\Aset_t(\Game,s) = B^{-1}_t\pi^\Aset_t(\Game,s)$ satisfies
\begin{gather*} 
\esssup_{\tau\in\Strat^\Aset_t(s)} \essinf_{\si\in\Strat^{-\Aset}_t(s)} \EMM\Big( \widehat V^\Aset_{\stp (s)}
(\tau,\si)\,\big|\,\Filt_t\Big) \leq \widehat\pi^\Aset_t(\Game,s) \leq \essinf_{\si\in\Strat^{-\Aset}_t(s)} \esssup_{\tau\in\Strat^\Aset_t(s)} \EMM\Big( \widehat V^\Aset_{\stp (s)} (\tau,\si)\,\big|\,\Filt_t\Big)
\end{gather*}
where $\Qrob$ is the unique martingale measure for the market model $(B,S)$.
\end{theorem}

\begin{proof}
We will prove that the upper bound holds if and only if no issuer's arbitrage exists. The statement regarding the lower bound can be established by applying analogous arguments to the game contingent claim $\MGCC(\widetilde \Game)$ with the payoffs defined by $\widetilde V^{-\Aset}(\tau,\si) = -V^\Aset (\tau,\si)$. Note that the upper bound 
can be rewritten as
\begin{gather}\label{eqro102}
\widehat\pi^\Aset_t(\Game,s)
\leq \essinf_{\si\in\Strat^{-\Aset}_t(s)} U^\si_t(\tau ).
\end{gather}
\noindent {\it First step.}  We first show that if there is no issuer's arbitrage at time $t$ then \eqref{eqro102} holds. From Proposition \ref{propro01}, we deduce that for any fixed $\si\in \Strat^{-\Aset}$, there exists an issuer's super-hedging strategy $(\phi^\si (\tau ),\si)$ on $[t,T]$ with the discounted value at time $t$ equal to $U^\si_t(\tau )$. Therefore, by Proposition \ref{proprm71}, if no issuer's arbitrage exists at time $t$, then $\widehat\pi^\Aset_t(\Game,s) \leq U^\si_t(\tau )$. The upper bound \eqref{eqro102} follows, as we may apply this argument to every $\si\in\Strat^{-\Aset}_t(s)$.

\noindent {\it Second step.} We will now show that if an issuer's arbitrage exists at time $t$ then \eqref{eqro102} fails to hold.
Assume that $(\phi^\si,\si)$ is an issuer's arbitrage, in the sense of Definition \ref{defrm03}. Then there exists an event $A \in \Filt_t$ with a positive probability $\Prob (A)$ (or, equivalently, $\Qrob (A)$) such that, on the event $A$,
\[
\widehat \Val_t(\phi^\si) < \widehat\pi^\Aset_t(\Game,s),\quad
\widehat \Val_{\stp (s)}(\phi^\si) \geq \widehat V^\Aset_{\stp (s)}(\tau,\si),\quad \forall\,\tau\in\Strat^\Aset_t(s).
\]
Since $\widehat \Val(\phi^\si)$ is a  $(\Qrob , \mathbb{F})$-martingale, from the optional sampling theorem we obtain, on the event $A$,
\begin{align*}
\widehat \Val_t(\phi^\si) = \EMM\Big( \widehat \Val_{\stp (s)}(\phi^\si)\,\big|\,\Filt_t\Big) \geq \EMM\Big( \widehat V^\Aset_{\stp (s)}(\tau,\si)\,\big|\,\Filt_t\Big),\quad \forall\,\tau\in\Strat^\Aset_t(s).
\end{align*}
Hence, by the definition of the essential supremum, we obtain on $A$
\begin{align*}
\widehat\pi^\Aset_t(\Game,s) > \widehat \Val_t(\phi^\si) \geq \esssup_{\tau\in\Strat^\Aset_t(s)} \EMM\Big( \widehat V^\Aset_{\stp (s)} (\tau,\si)\,\big|\,\Filt_t\Big)=U^\si_t(\tau ) \geq \essinf_{\si\in\Strat^{-\Aset}_t(s)} U^\si_t(\tau ).
\end{align*}
Consequently,
$$
\Qrob\Big(\widehat\pi^\Aset_t(\Game,s) >  \essinf_{\si\in\Strat^{-\Aset}_t(s)} U^\si_t(\tau )\Big)>0
$$
and thus \eqref{eqro102} fails to hold.
\end{proof}


\subsection{Continuous-Time Case}\label{sec7f}

In this section, we consider an arbitrage-free and complete market model with the unique martingale measure $\mathbb{Q}$,
for instance, the multi-dimensional Black-Scholes model (see Karatzas and Shreve \cite{Karatzas1}). We emphasise that, in the continuous-time setup, the filtration $\FF$ is assumed to satisfy the usual conditions of $\Qrob$-completeness and right-continuity. For $i=1,\ldots,d$, the discounted stock price processes $\wh{S}^1, \dots , \wh{S}^d$ are $(\Qrob , \mathbb{F})$-local martingales.
Let us denote by ${\cal L}(\wh{S})$ the class of all $\mathbb{R}^d$-valued, $\FF$-predictable processes that are $\wh{S}$-integrable on $[0,T]$. More precisely, we identify any processes $\phi$ and $\wt \phi $ from ${\cal L}(\wh{S})$ such that
\[
\int_{(0,t]} \phi_u \, d \wh{S}_u = \int_{(0,t]} \wt \phi_u \, d \wh{S}_u , \quad \forall \, t \in [0,T],
\]
and we write $\phi \stackrel{\wh{S}} = \wt \phi $ in that case. Due to the local properties of the It\^o stochastic
integral, for any stopping time $\rho \in {\cal T}_{[0,T]}$ and any process $\phi \in {\cal L}(\wh{S})$,
if $\int_{(0, \cdot ]}\phi_u \, d\wh{S}_u = 0$ on $\cll 0, \rho \crr $, meaning that
\[
\int_{(0,t \wedge \rho ]} \phi_u \, d \wh{S}_u = 0 , \quad \forall \, t \in [0,T],
\]
then $\phi \stackrel{\wh{S}} = 0$ on  $\cll 0, \rho \crr $. When searching for super-hedging strategies,
we need to modify accordingly Definition \ref{defpredi}  of an $\Hist$-predictable mapping.

\begin{definition}  {\rm
A mapping $\phi \map{\Hist\times[0,T]\times\Omega}{\mathbb{R}^d}$ is \emph{$(\wh{S},\Hist )$-predictable} if
$\phi (h) \in {\cal L}(\wh{S})$ for every $h \in \Hist$ and, for every $h,h'\in\Hist$, the equality $\phi (h) \stackrel{\wh{S}} =\phi (h')$ holds on $\cll 0, \rho \crr$ where $\rho = \rho (h,h') $ is given by~\eqref{veqsg40}.}
\end{definition}

Throughout Section \ref{sec7f}, we work under the following standing assumption, which complements Assumption \ref{assrm09}.

\begin{assumption} \label{asm09} {\rm
The discounted stock price $\wh{S}= (\wh{S}^1, \dots , \wh{S}^d)$ has the {\it predictable representation property} with
respect to the filtration $\mathbb{F}$, that is, for any  $(\Qrob , \FF)$-martingale $M$ there exists a process
$\phi \in {\cal L}(\wh{S})$ such that $M = \int_{(0, \cdot ]}\phi_u \, d\wh{S}_u$.}
\end{assumption}


\subsubsection{Right-Continuous Games} \label{secstrong}

Recall that the $\Hist^{\si }$-predictable  Snell envelope $U^\si(\tau )$ is a  $(\Qrob , \mathbb{F})$-supermartingale for any fixed $(\tau,\si)\in\Strat$. To be able to apply the Doob-Meyer decomposition theorem, it suffices to show that $U^\si(\tau)$ is an
RCLL supermartingale of class (D). In view of part (ii) in Proposition \ref{proprm50}, it thus remains to show that $U^\si(\tau)$
admits an RCLL modification. To this end, we recall the classic Doob regularisation theorem (see, e.g., Theorem 6.27 in Kallenberg \cite{Kallenberg}).


\begin{theorem} \label{thmsg07}
A $(\Qrob ,\FF )$-supermartingale $U$ admits an RCLL modification if and only if the function $\EMM (U_t),\, t \in [0,T]$ is right-continuous.
\end{theorem}

 By applying this result to our case, we see that the supermartingale $U^\si(\tau)$ has an RCLL modification whenever the expected value $\EMM(U^\si_t(\tau))$ is right-continuous. Note also that, from the proof of part (i) in Lemma \ref{lemm50x}, we deduce that
\begin{gather}\label{eqsg102}
\EMM(U^\si_t(\tau)) =\sup_{\tau' \in \Strat^\Aset_t(\tau,\si)} \EMM\Big(\widehat V^\Aset_{\stp(\tau',\si)} (\tau',\si) \Big).
\end{gather}
The right-continuity of the function given by \eqref{eqsg102} is not guaranteed, in general.
Therefore, we need to impose  additional regularity conditions on the game $\Game$ under the martingale measure $\mathbb{Q}$.
To be more specific, in this subsection, we work under the following assumption.

\begin{assumption} \label{asssg20} {\rm
A continuous-time game $\Game$ satisfies the following  {\it right-continuity condition}:
for any $\Aset\subseteq\Mset$ and arbitrary $s=(\tau,\si)\in\Strat^{-\Aset}\times\Strat^\Aset$,
the following equalities hold for every $t\in[0,T]$,}
\begin{align}
\lim_{u\downarrow t} \sup_{\tau' \in \Strat^\Aset_u(\tau,\si)} \EMM\Big(\widehat V^\Aset_{\stp(\tau',\si)}  (\tau',\si) \Big) &= \sup_{\tau' \in \Strat^\Aset_t(\tau,\si)} \EMM\Big(\widehat V^\Aset_{\stp(\tau',\si)} (\tau',\si) \Big),\\
\lim_{u\downarrow t} \inf_{\si' \in \Strat^{-\Aset}_u(\tau,\si)} \EMM\Big(\widehat V^\Aset_{\stp(\tau,\si')}  (\tau,\si') \Big) &= \inf_{\si' \in \Strat^{-\Aset}_t(\tau,\si)} \EMM\Big(\widehat V^\Aset_{\stp(\tau,\si')} (\tau,\si') \Big).
\end{align}
\end{assumption}

Any game $\Game $ satisfying Assumption \ref{asssg20} is referred to as a \emph{right-continuous game}.
The motivation for this name stems from the fact that Assumption \ref{asssg20} allows us to establish the existence of an RCLL modification of the $\Hist^\si$-predictable Snell envelope of the first kind, as the following result shows.

\begin{corollary}\label{propsg22}
Assume that a game $\Game $ is right-continuous. Then, for any fixed $(\tau,\si)\in\Strat$, the process $U^\si(\tau)$ has an RCLL modification which is a $(\mathbb{Q},\FF)$-supermartingale of class~$(D)$.
\end{corollary}

\begin{proof}
From Assumption \ref{asssg20}, it follows that the function given by the right-hand side in \eqref{eqsg102} is right-continuous.
Hence the assertion is an immediate consequence of Proposition \ref{proprm50} and Theorem~\ref{thmsg07}.
\end{proof}

\begin{lemma}\label{propsg04}
Suppose that the mapping $\wh{M}\map{\Hist\times[0,T]\times\Omega}{\R}$ satisfies the following properties: \hfill \break
(i) for any $s\in\Strat$, the process $\wh{M}(s)=\wh{M}(h(s))$ is an RCLL $(\mathbb{Q},\FF)$-martingale, \hfill \break
(ii) $\wh{M}$ is $\Hist$-predictable.\\
Then there exists an $(\wh{S},\Hist )$-predictable mapping $\phi $ such that, for every $s \in \Strat$, the discounted wealth
 of a trading strategy $\phi (s)$ satisfies  $\widehat Z(\phi (s))=\wh{M}(s)$.
The statement is still valid if we replace $\Hist$ and $[0,T]$ by $\Hist^\si$ and $[t,T]$, respectively.
\end{lemma}

\begin{proof}
For any fixed $s\in\Strat$, by the postulated predictable representation property of $\wh{S}$ with respect to $\FF$
(see Assumption \ref{asm09}), there exists an $\FF$-predictable, $\R^d$-valued, integrable process $\phi (s)$ such that
\[
\wh{M}_t(s)=\wh{M}_0(s)+\sum_{j=1}^d \int_0^t \phi^j_u(s)\, d\wh{S}^j_u,\quad  t \in[0,T].
\]
It is routine to check that $\phi(s)$ is indeed an admissible trading strategy with $\widehat \Val(\phi(s))=\wh{M}(s)$.
For any $s,s'\in\Strat$, we define the random time  (see \eqref{veqsg40})
\begin{gather}\label{eqsg40}
\rho (s,s') := \inf \big\{ t \in [0,T] : h_t(s) \neq h_t(s') \big\} \wedge T.
\end{gather}
Then $\rho (s,s') $ is an $\FF$-stopping time and, for every $t \in [0,T]$, the event $\{h_{[0,t)}(s)=h_{[0,t)}(s')\}$
coincides with $\{t \leq \rho (s,s') \}$. By the $\Hist$-predictability of the mapping $\wh{M}$, the martingale $\wh{M}(s)-\wh{M}(s')$ vanishes on $\cll 0, \rho \crr $, and thus, one again we conclude
that $\phi (s) - \phi (s')  \stackrel{\wh{S}}=  0$ on $\cll 0, \rho \crr $. This shows that the mapping $\phi $ is $(\wh{S},\Hist )$-predictable.
\end{proof}

The following proposition, which holds under Assumption \ref{asssg20}, is the continuous-time analogue of Proposition \ref{propro01}. We fix here $t$, and thus all mappings are considered on $[t,T]$, rather than $[0,T]$.

\begin{proposition} \label{propsg21}
Fix $\si\in\Strat^{-\Aset}$ and $t\in[0,T]$.  There exists a mapping $\phi^\si  \map{\Hist^{\si } \times[t,T]\times\Omega}{\R}$, which is $(\wh{S},\Hist^\si )$-predictable on $[t,T]$, such that for every $\tau\in\Strat^{\Aset}$, on the event $\{\stp(\tau,\si)\geq t\}$, the pair $(\phi^\si (\tau ),\si)$ is an issuer's super-hedging strategy and its discounted wealth
at time $t$ is equal to the value $U^\si_t(\tau)$ of the $\Hist^{\si }$-predictable Snell envelope of the first kind, so that
\begin{gather}
\widehat Z_{t}(\phi^\si(\tau))=U^\si_t(\tau)=\esssup_{\tau'\in\Strat^\Aset_t(\tau,\si)} \EMM \Big(\widehat V^{\Aset}_{\stp(\tau',\si)}(\tau',\si) \,\Big|\,\Filt_t\Big),\label{eqsg211}
\end{gather}
and
\begin{gather}
\widehat Z_{\stp (s)}(\phi^\si(\tau))\geq U^\si_{\stp(s)}(\tau) \geq  \widehat V^{\Aset}_{\stp(s)}(\tau,\si) . \label{eqsg212}
\end{gather}
\end{proposition}

\begin{proof}
The proof is analogous to the proof of Proposition \ref{propro01}.
From Proposition \ref{propsg22}, we know that, for any $\tau \in \Strat^\Aset $, the Snell envelope $U^\si(\tau)$ is an RCLL $\Qrob$-supermartingale of class (D). In view of the Doob-Meyer decomposition theorem, there exists a continuous, uniformly integrable $(\Qrob ,\FF)$-martingale $M^\si(\tau)$ and an increasing, $\FF$-predictable process $A^\si(\tau)$ satisfying $U^\si(\tau) = M^\si(\tau) - A^\si(\tau),\,  A^\si_0(\tau)=0$ and this decomposition is unique.
We define a candidate for the discounted wealth process by setting
$\widehat M^\si_u(\tau)= M_u^\si(\tau)- A_t^\si(\tau)$ for all $u \in [t,T]$.
It clear that, for any fixed $\tau \in \Strat^\Aset $, the process $\widehat Z^\si(\tau )$ is a uniformly
integrable $(\Qrob ,\FF)$-martingale.

We will show that the mapping $\widehat M^\si$ is $\Hist^\si$-predictable on $[t,T]$. To this end, we take arbitrary $s=(\tau,\si),\, s'=(\tau',\si)$ and we define the $\FF$-stopping time $\rho = \rho (\tau ,\tau')$ using equation \eqref{tautau}.
Since $U^\si $ is $\Hist^\si$-predictable, we have that $U^\si(\tau)=U^\si(\tau')$ on
$\cll 0 , \rho \crr $. From the uniqueness of the Doob-Meyer decomposition,  we thus obtain
$M^\si(\tau) = M^\si(\tau ')$ and $A^\si(\tau) = A^\si(\tau')$ on $\cll 0 , \rho \crr $.
We deduce easily that $\widehat M^\si(\tau ) = \widehat M^\si(\tau ')$
on $\cll t , \rho \crr $ and thus the $\Hist^\si$-predictability of $\wh{M}^\si $ on $[t,T]$ is established.
From Lemma \ref{propsg04}, we deduce the existence of a mapping $\phi^\si  \map{\Hist^{\si } \times[t,T]\times\Omega}{\R}$, which is $(\wh{S},\Hist^\si )$-predictable and such that the equality $\widehat Z(\phi^\si (\tau ) )=\widehat{M}^\si (\tau )$ holds
on $[t,T]$.

It now remains to check that \eqref{eqsg211} and \eqref{eqsg212} are valid. For any $\FF$-stopping time $\stp (s)$ taking values in $[t,T]$, we obtain, on the event $\{ \stp (s) \geq t \}$,
\[
\widehat Z_{\stp (s)}(\phi^\si(\tau)) =\widehat M_{\stp (s)}(\phi^\si(\tau)) = M^\si_{\stp (s)}(\tau)-A^\si_t(\tau) = U^\si_{\stp (s)}(\tau)+A^\si_{\stp (s)}(\tau)-A^\si_t(\tau)\geq U^\si_{\stp (s)}(\tau)
\]
with equality holding on the event $\{ \stp (s)=t \}$. This completes the proof of the theorem.
\end{proof}


\subsubsection{General Games} \label{seci7f}

Assumption \ref{asssg20} is in fact too strong, since it is easy to produce examples
of games for which it fails to hold. Therefore, this assumption will be now relaxed, and we will
show that a counterpart of Corollary \ref{propsg22} can still be established (see Corollary \ref{lempp9}).
To this end, we will introduce an alternative definition of the $\Hist^\si $-predictable Snell envelope
(see Definition \ref{defirm44x}), so that we need first to analyze some basic relationships between the Snell envelopes $U^\si $ and $\wt U^\si $.  From Section \ref{sec3.3x}, we already know that the processes $U^\si(\tau)$ and $\wt U^\si(\tau)$  are  $(\Qrob , \FF)$-supermartingales, and thus their expected values are decreasing functions of time.

\begin{lemma} \label{lemhh31}
(i) For any $\tau \in\Strat^{\Aset} $ and $t<u$, we have
\begin{gather}  \label{asas}
U^\si_t(\tau) \geq \wt U^\si_t(\tau) \geq \EMM \big(U^\si_u(\tau)\,\big|\,\Filt_t\big)
\geq \EMM \big(\wt U^\si_u(\tau)\,\big|\,\Filt_t\big).
\end{gather}
(ii) The following set $\Isi (\tau)$ is countable
\[
\Isi (\tau) := \big\{ t\in[0,T] : \EMM(U^\si_t(\tau)) > \EMM(\wt U^\si_t(\tau)) \big\}.
\]
(iii) If $t\notin \Isi (\tau)$, then $U^\si_t(\tau) = \wt U^\si_t(\tau)$.
\end{lemma}

\begin{proof}
(i) The asserted inequalities follow from part (i) in Lemma \ref{lemm50x} and the following obvious inclusions $\wt \Strat^\Aset_u(\tau,\si)\subseteq \Strat^\Aset_u(\tau,\si)\subseteq \wt \Strat^\Aset_t(\tau,\si)\subseteq \Strat^\Aset_t(\tau,\si)$.

\noindent (ii) We start by noting that the inequality $\EMM(U^\si_t(\tau)) \geq \EMM(\wt U^\si_t(\tau))$ holds for all $t$.
For any $(\tau,\si)$, we define the set  $\Isi (\tau)$ of dates for which the inequality is strict.
In view of \eqref{asas}, the decreasing functions $\EMM(U^\si_t(\tau))$ and $\EMM(\wt U^\si_t(\tau))$ have the same
right-hand and left-hand limits for every $t$. Therefore, they are equal at any $t$ for which at least one of them
(and thus both) are continuous. To conclude the proof of part (ii), we observe that the set of points of discontinuity of a decreasing function is countable.

\noindent (iii) We observe that if $t\notin \Isi (\tau)$ then $\EMM(U^\si_t(\tau)) = \EMM(\wt U^\si_t(\tau))$.
Hence the assertion is a direct consequence of the inequality $U^\si_t(\tau) \geq \wt U^\si_t(\tau)$ (see part (i)).
\end{proof}


We define the mappings $B^\si$ and $\wt B^\si$, which capture the cumulative differences between $U^\si$ and $\wt U^\si$, by setting, for all $t \in [0,T]$ and $\tau \in \Strat^{\Aset }$,
\begin{gather*}
B^\si_t(\tau):=\sum_{u \in \Isi_t (\tau )} U^\si_u(\tau)-\wt U^\si_u(\tau),\\
\wt B^\si_t(\tau):=\sum_{u \in \wtIsi_t (\tau )} U^\si_u(\tau)-\wt U^\si_u(\tau),
\end{gather*}
where we denote $\Isi_t (\tau )= [0,t) \cap \Isi (\tau)$ and $\wtIsi_t (\tau ) = [0,t] \cap \Isi (\tau)$
and we apply the usual convention that the empty sum equals zero.
In view of part (iii) in Lemma \ref{lemhh31}, we have $B^\si_t(\tau) = \wt B^\si_t(\tau)$ for every $t\notin \Isi (\tau)$.

\begin{lemma} \label{lembxb}
(i) The processes $B^\si(\tau)$ and $\wt B^\si(\tau)$ are non-negative and increasing.  \\
(ii) The mapping $B^\si$ ($\wt B^\si$, resp.) is $\Hist^\si$-predictable ($\Hist^\si$-adapted, resp.).\\
(iii) For any $t\leq u<v\leq T$, we have the inequality
\[
\EMM \Big(\wt B^\si_v(\tau)-\wt B^\si_u(\tau)\,\big|\,\Filt_t\Big) \leq \EMM \Big(\wt U^\si_u(\tau)-\wt U^\si_v(\tau)\,\big|\,\Filt_t\Big).
\]
\end{lemma}

\begin{proof}
(i) By part (ii) in Lemma \ref{lemhh31}, the set $\Isi (\tau)$ is countable and thus the inequality $U^\si_t(\tau)\geq \wt U^\si_t(\tau)$ holds for all $t\in \Isi (\tau)$, with probability 1. Therefore, $B^\si(\tau)$ and $\wt B^\si(\tau)$ are non-negative increasing processes.

\noindent (ii) To show that $B^\si $ is $\Hist^\si$-predictable, it suffices to check that $B^\si(\tau) = B^\si (\tau')$ on $[\![ 0, \rho ]\!]$ where $\rho = \rho (\tau ,\tau')$ is given by \eqref{tautau}. Since, by part (iii) in Lemma \ref{lemhh31},  equality $U^\si_t(\tau) = \wt U^\si_t(\tau)$ holds for every $t\notin \Isi (\tau)$, we obtain
\begin{gather*}
B^\si_t(\tau) 
 =\sum_{u \in \Isi_t (\tau)\cup \Isi_t (\tau')} U^\si_u(\tau)-\wt U^\si_u(\tau),\\
B^\si_t(\tau') 
=\sum_{u \in \Isi_t (\tau)\cup \Isi_t (\tau')} U^\si_u(\tau')-\wt U^\si_u(\tau').
\end{gather*}
Both sums are now over same countable index set, and thus the $\Hist^\si$-predictability of $B^\si$ follows from
the definition of $\Isi_t (\tau )$ and $\Isi_t (\tau' )$, since the $\Hist^\si$-adaptedness of $U^\si$ and $\wt U^\si$
implies that $U^\si(\tau) = U^\si(\tau')$ and $\wt U^\si(\tau) = \wt U^\si(\tau')$ on $[\![ 0, \rho )\!)$.
The $\Hist^\si$-adaptedness of $\wt B^\si $ can be shown using similar arguments.

\noindent (iii) From parts (i) and (ii) in Lemma \ref{lemhh31}, the following set of countable stochastic open intervals
\[
\Big\{ \Big(\EMM \big(\wt U^\si_w(\tau)\,\big|\,\Filt_t\big),\ \EMM \big(U^\si_w(\tau)\,\big|\,\Filt_t\big) \Big)
: w \in(u,v]\cap \Isi(\tau) \Big\}
\]
consists of pairwise disjoint stochastic intervals lying within
\[
 \Big[\EMM \big(\wt U^\si_v(\tau)\,\big|\,\Filt_t\big),\ \EMM \big( \wt U^\si_u(\tau)\,\big|\,\Filt_t\big) \Big],
\]
almost surely. Consequently, the non-negative random variable $\wt B^\si_v(\tau)-\wt B^\si_u(\tau)$ satisfies
\begin{align*}
\EMM \Big(\wt B^\si_v(\tau)-\wt B^\si_u(\tau)\,\big|\,\Filt_t\Big) &= \sum_{w \in(u,v]\cap \Isi (\tau)} \EMM \Big(U^\si_w (\tau)\,\big|\,\Filt_t\Big)-\EMM \Big(\wt U^\si_w (\tau)\,\big|\,\Filt_t\Big)\\
&\leq \EMM \Big( \wt U^\si_u(\tau) - \wt U^\si_v(\tau)\,\big|\,\Filt_t\Big)
\end{align*}
where the equality follows from Fubini's theorem
The inequality $\EMM \big(\wt B^\si_T(\tau)\big) < \infty$ now follows, and thus the process $\wt B^\si (\tau)$
is $\Qrob$-integrable.
\end{proof}

The following lemma establishes right-continuity conditions, which can be seen as an alternative to Assumption~\ref{asssg20}.

\begin{lemma} \label{propsg20}
For every $\Aset\subseteq\Mset$ and $s=(\tau,\si)\in\Strat^{-\Aset}\times\Strat^\Aset$,
the following equalities hold for every $t\in[0,T]$,
\begin{align}
\lim_{u\downarrow t} \sup_{\tau' \in \wt\Strat^\Aset_u(\tau,\si)} \EMM\Big(\widehat V^\Aset_{\stp(\tau',\si)}  (\tau',\si) \Big) &= \sup_{\tau' \in \wt\Strat^\Aset_t(\tau,\si)} \EMM\Big(\widehat V^\Aset_{\stp(\tau',\si)} (\tau',\si) \Big), \label{rrtt} \\
\lim_{u\downarrow t} \inf_{\si' \in \wt\Strat^{-\Aset}_u(\tau,\si)} \EMM\Big(\widehat V^\Aset_{\stp(\tau,\si')}  (\tau,\si') \Big) &= \inf_{\si' \in \wt\Strat^{-\Aset}_t(\tau,\si)} \EMM\Big(\widehat V^\Aset_{\stp(\tau,\si')} (\tau,\si') \Big).
\end{align}
\end{lemma}

\begin{proof}
By symmetry, it suffices to only prove the first equality. From part (i) in Proposition \ref{prprm50}, we have the equality
\[
\EMM\big(\wt U^\si_t(\tau)\big)=\sup_{\tau' \in \wt\Strat^\Aset_t(\tau,\si)} \EMM\Big(\widehat V^\Aset_{\stp(\tau',\si)} (\tau',\si) \Big),
\]
so that \eqref{rrtt} is equivalent to the right-continuity of the function $\EMM\big(\wt U^\si_t(\tau)\big),\, t \in [0,T]$.
Since for $t<w$, we have $\wt\Strat^\Aset_w(\tau,\si) \subseteq \wt\Strat^\Aset_t(\tau,\si)$, and thus the function $\EMM\big(\wt U^\si_t(\tau)\big)$ is decreasing. Hence the right-hand limit $\lim_{w\downarrow t} \EMM\big(\wt U^\si_w(\tau)\big)$ is well-defined and bounded above by $\EMM\big(\wt U^\si_t(\tau)\big)$.

Clearly, to establish equality \eqref{rrtt}, that is, the right-continuity of the function $\EMM\big(\wt U^\si(\tau)\big)$, it is enough to show that $\EMM\big(\wt U^\si_t(\tau)\big)\leq \lim_{w\downarrow t}\EMM\big(\wt U^\si_w(\tau)\big)$ for all $t \in [0,T]$. For this purpose, we fix $\wt \tau\in \wt\Strat^\Aset_t(\tau,\si)$ and we consider the stopping time $\rho $ given by
\[
\rho:= \rho(\tau,\wt \tau)=\inf \big\{ u\in[0,T]: h_u(\tau,\si) \neq h_u(\si,\wt \tau) \big\} \wedge T.
\]
Since $\wt \tau\in \wt\Strat^\Aset_t(\tau,\si)$, we must have $h_{[0,t]}(\tau,\si)=h_{[0,t]}(\wt \tau,\si)$ and thus $\rho >t$. Consider the process $\big(\wt U^\si_w(\tau)-\wt V^\Aset_w(\wt \tau,\si)\big) \I_{\{w<\rho\}}$, for $w \in (t,T]$,
where
\[
\wt V^\Aset_w(\wt \tau,\si) := \EMM\Big(\widehat V^\Aset_{\stp(\wt \tau,\si)} (\wt\tau,\si) \,\Big|\, \Filt_w \Big).
\]
Assumption \ref{assdv01} and the finiteness of the sets of actions yield $\{w < \rho\} = \{h_{[0,w]}(\tau,\si) = h_{[0,w]}(\si,\wt \tau)\}$. Consequently, 
we obtain
\[
\wt V^\Aset_w (\wt \tau,\si)\I_{\{w<\rho\}}
\leq \esssup_{\tau' \in \wt\Strat^\Aset_w(\tau,\si)}\wt V^\Aset_w  (\tau',\si) \I_{\{w<\rho\}}=\wt U^\si_w(\tau) \I_{\{w<\rho\}},
\]
which in turn implies that
\begin{gather}\label{eqsg2002}
\big(\wt U^\si_w(\tau)-\wt V^\Aset_w(\wt\tau,\si)\big) \I_{\{w<\rho\}} \geq 0.
\end{gather}
Observe that $\wt V^\Aset (\wt\tau,\si)$ is a $(\Qrob,\FF)$-martingale and thus, by Theorem \ref{thmsg07},
it admits an RCLL modification. We also know that $\wt U^\si (\tau)$ is a $(\Qrob,\FF)$-supermartingale (see part (iii) of Proposition \ref{prprm50}), but its right-continuity was not yet established. Doob's regularization theorem (see, e.g., Theorem 6.27 in Kallenberg \cite{Kallenberg})
states that for any submartingale $X$ on $\mathbb{R}_+$ with restriction $Y$ on the set on positive rational numbers $\mathcal{Q}_+$, the process $Y^+_{t} :=\lim_{w\downarrow t, w\in\mathcal{Q}} Y_w$ is an RCLL process outside some fixed $\Qrob$-null set. By slightly modifying the proof of this theorem we may and do assume, without loss
of generality, that the countable set
\[
I(\tau ,\si ) := \big\{ t\in[0,T] : \EMM(U^\si_t(\tau)) > \lim_{w\downarrow t} \EMM(\wt U^\si_w(\tau)) \big\},
\]
is included in the countable, dense subset $\wt{\mathcal{Q}}_+$ of $\mathbb{R}_+$. So we may define, for all $t \in \wt{\mathcal{Q}}_+$,
\[
\wt U^\si_{t+}(\tau) := \lim_{w\downarrow t, w\in \wt{\mathcal{Q}}_+} \wt U^\si_w(\tau).
\]
One the one hand, the dominated convergence theorem (recall that we work under Assumption \ref{assrm09}) yields, for all $t \in \wt{\mathcal{Q}}_+$,
\begin{gather}\label{eqsg2003}
\lim_{w\downarrow t} \EMM\big(\wt U^\si_w(\tau)\big) =\lim_{w\downarrow t, w\in \wt{\mathcal{Q}}_+} \EMM\big(\wt U^\si_w(\tau)\big)=\EMM\big(\wt U^\si_{t+}(\tau)\big).
\end{gather}
On the other hand, by applying the dominated convergence theorem to \eqref{eqsg2002}, we obtain
\begin{align} \label{eqsg2004}
0\leq \lim_{w\downarrow t, w\in \wt{\mathcal{Q}}_+} \, \EMM\Big(\big(\wt U^\si_w(\tau)-\wt V^\Aset_w(\wt \tau,\si)\big) \I_{\{w<\rho\}} \Big) = \EMM\big(\wt U^\si_{t+}(\tau)-\wt V^\Aset_t(\wt \tau,\si)\big)
\end{align}
since the process $\I_{\{w<\rho\}},\, w \in [t,T]$, is bounded, with right-continuous sample paths, and $\I_{\{t<\rho\}}=1$.
By combining \eqref{eqsg2003} with \eqref{eqsg2004}, we conclude that
\begin{align}\label{eqsg2005}
\lim_{w\downarrow t} \EMM\big(\wt U^\si_w(\tau)\big)=\EMM\big(\wt U^\si_{t+}(\tau)\big)
\geq \EMM\big(\wt V^\Aset_t(\wt \tau,\si)\big)=\EMM\big(\widehat V^\Aset_{\stp(\wt \tau,\si)} (\wt \tau,\si) \big).
\end{align}
Since \eqref{eqsg2005} holds for all $\wt \tau\in \wt\Strat^\Aset_t(\tau,\si)$, we must have, for all $t \in \wt{\mathcal{Q}}_+$,
\[
\lim_{w\downarrow t} \EMM\big(\wt U^\si_w(\tau)\big)\geq \sup_{\tau' \in \wt\Strat^\Aset_t(\tau,\si)} \EMM\Big(\widehat V^\Aset_{\stp(\tau',\si)} (\tau',\si) \Big)=\EMM\big(\wt U^\si_t(\tau)\big).
\]
In view of the definition of the set $I(\tau ,\si )\subset \wt{\mathcal{Q}}_+$, this completes the proof of the lemma.
\end{proof}

%
%
%

\begin{corollary} \label{lempp9}
For any $(\tau,\si)\in\Strat^{-\Aset}\times\Strat^\Aset$, we have that \\
(i) the process $\wt U^\si(\tau)$ has an RCLL modification,
which is a $(\Qrob , \FF)$-supermartingale of class (D). \\
(ii) the process $\wt B^\si(\tau)$ has an RCLL modification,
which is a $(\Qrob , \FF)$-submartingale of class (D).
\end{corollary}

\begin{proof} (i) In view of part (iii) in Proposition \ref{prprm50}, the process $\wt U^\si(\tau)$ is a $(\Qrob , \FF)$-supermartingale. Moreover, by equality \eqref{rrtt} in Lemma \ref{propsg20}, its expected value
$\EMM \big(\wt U^\si_t(\tau)\big)$
is a right-continuous function. The existence of an RCLL modification of $\wt U^\si(\tau)$ thus follows from Theorem \ref{thmsg07}.
To show that it is $(\Qrob , \FF)$-supermartingale of class (D), we use part (iii) in Proposition \ref{prprm50} and we argue as is the proof of part (ii) in Proposition \ref{proprm50}.

\noindent (ii)  To show that $\wt B^\si(\tau)$ has an RCLL modification, it suffices to check that
the function $\EMM \big(\wt B^\si_t(\tau)\big)$ is right-continuous or, more explicitly,
\[
\lim_{u\downarrow t} \sum_{v\in(t,u]\cap \Isi (\tau) } \EMM (U^\si_v(\tau)) -\EMM (\wt U^\si_v(\tau)) = 0.
\]
Since the functions $\EMM \big(U^\si_t(\tau)\big)\geq \EMM \big(\wt U^\si_t(\tau)\big)$ are decreasing and the
equality $\EMM \big(U^\si_t(\tau)\big)=\EMM \big(\wt U^\si_t(\tau)\big)$ holds for all $t \notin \Isi (\tau )$, we obtain
\[
0 \leq \sum_{v \in (t,u]\cap \Isi (\tau)} \EMM (U^\si_v(\tau)) -\EMM (\wt U^\si_v(\tau))
\leq  \EMM (\wt U^\si_t(\tau)) -\EMM (\wt U^\si_u(\tau)),
\]
where the difference tends to 0 as $u\downarrow t$ since, from part (i), the function $\EMM \big(\wt U^\si_t(\tau)\big)$ is right-continuous.
\end{proof}

\subsubsection{$\Hist^{\si }$-Predictable Snell Envelope of the Second Kind}

We are in a position to introduce another version of the concept of the $\Hist^\si$-predictable Snell envelope, which
will serve as a convenient tool when searching for the issuer's super-hedging strategy.

\begin{definition}\label{defirm44x}{\rm
For a combined tranche $\MGCC^\Aset(\Game)$ and a fixed $\si \in \Strat^{-\Aset }$, the \emph{$\Hist^\si$-predictable Snell envelope of the second kind} $\bbUU^\si \map{\Hist^\si \times[0,T]\times\Omega}{\R}$ is given by
\begin{gather}  \label{eqhh30}
\bbUU^\si_t(\tau) :=  \wt U^\si_t(\tau)+ \wt{B}^\si_t(\tau) = U^\si_t(\tau) + B^\si_{t}(\tau)
\end{gather}
where the second equality follows from part (iii) in Lemma \ref{lemhh31}.}
\end{definition}

Intuitively, the difference $ \wt{B}^\si_{t}(\tau) = \bbUU^\si_t(\tau) - \wt{U}^\si_t(\tau)$ represents the cumulative adjustments to the process $\wt{U}^\si(\tau)$ that are needed to ensure that the inequality $\bbUU^\si_t(\tau) \geq U^\si_t(\tau)$ holds for all $t$. It is fair to say, however, that $\bbUU^\si(\tau)$ does not enjoy as nice interpretation as $U^\si(\tau)$, so it should merely be seen as a technical tool to derive the desired results, Propositions \ref{priopog21} and \ref{propog21}.

\begin{proposition} \label{priopog21}
For any $\si\in\Strat^{-\Aset}$: \\
(i) the mapping $\bbUU^\si$ is an $\Hist^\si$-predictable, \\
(ii) for any $\tau \in \Strat^\Aset$, the process $\bbUU^\si(\tau)$ is an RCLL $(\Qrob , \FF)$-supermartingale of class (D).
\end{proposition}

\begin{proof} (i)  It suffices to observe that the second equality in \eqref{eqhh30} yields $\bbUU^\si = U^\si + B^\si$ where $U^\si$ and $B^\si$ are $\Hist^\si$-predictable by Lemmas \ref{lmcc7a} and \ref{lembxb}, respectively.

\noindent (iii) The first equality in \eqref{eqhh30} reads $\bbUU^\si(\tau) =  \wt U^\si(\tau)+ \wt B^\si(\tau)$.
Hence, from part (ii) in Lemma \ref{lembxb}, the process $\bbUU^\si(\tau)$ is a $(\Qrob , \FF)$-supermartingale.
Next, from Corollary \ref{lempp9}, it follows that $\bbUU^\si(\tau)$ has an RCLL modification and is of class (D).
\end{proof}

We will now prove the existence of a super-hedging strategy on $[t,T]$. It is
worth noting that we would like to employ here the following arguments, for any fixed $t$: \hfill \break
(i) the process $\bbUU^\si_u(\tau),\, u \in [t,T]$ is an RCLL $(\Qrob , \FF)$-supermartingale of class (D),\hfill \break
(ii) we may and do assume, without loss of generality, that the equality $\bbUU^\si_t(\tau)= U^\si_t(\tau)$
holds, for instance, by replacing the process $\wt{B}^\si(\tau)$ by $\wt{B}^\si(\tau)\I_{[t,T]}$ in formula \eqref{eqhh30}, that is, by neglecting the past adjustments to $\bbUU^\si(\tau)$ on~$[0,t)$.

Formally, in view of the second equality in \eqref{eqhh30}, for a fixed $t \in [0,T]$, to make use of the above
 features (i) and (ii), it is enough to define the auxiliary process $\check U_u^\si (\tau ) = \bbUU^\si_u (\tau ) - B_t ,\, u \in [t,T]$, which obviously satisfies $\check{U}^\si_t (\tau ) = U^\si_t (\tau )$ and is an RCLL $(\Qrob , \FF)$-supermartingale of class (D).

\begin{proposition} \label{propog21}
Fix $\si\in\Strat^{-\Aset}$ and $t\in[0,T]$.  There exists a mapping $\phi^\si  \map{\Hist^{\si } \times[t,T]\times\Omega}{\R}$, which is $(\wh{S},\Hist^\si )$-predictable on $[t,T]$, such that for every $\tau\in\Strat^{\Aset}$, on the event $\{\stp(\tau,\si)\geq t\}$, the pair $(\phi^\si (\tau ),\si)$ is an issuer's super-hedging strategy and its discounted wealth
at time $t$ is equal to the value $\bbUU^\si_t(\tau)$ of the $\Hist^{\si }$-predictable Snell envelope of the second kind, so that
\begin{gather}  \label{xeqsg211}
\widehat Z_{t}(\phi^\si(\tau))=\check{U} ^\si_t(\tau)= U^\si_t(\tau) = \esssup_{\tau'\in\Strat^\Aset_t(\tau,\si)} \EMM \Big(\widehat V^{\Aset}_{\stp(\tau',\si)}(\tau',\si) \,\Big|\,\Filt_t\Big),
\end{gather}
and
\begin{gather} \label{xeqsg212}
\widehat Z_{\stp (s)}(\phi^\si(\tau))\geq \check{U}^\si_{\stp(s)}(\tau) \geq U^\si_{\stp(s)}(\tau) \geq \widehat V^{\Aset}_{\stp(s)}(\tau,\si) .
\end{gather}
\end{proposition}

\begin{proof}
The proof goes along the same lines as the proof of Proposition \ref{propsg21} with only minor modifications.
We may now apply the Doob-Meyer decomposition $\check{U}^\si(\tau)=\check{M}^\si(\tau)-\check{A}^\si(\tau)$, to obtain an $\Hist^\si$-predictable martingale mapping $\wh{M}^\si_u(\tau),\, u \in [t,T]$, which dominates $U^\si (\tau)$ on $[t,T]$, and which will play the role of
the wealth process of a super-hedging portfolio on $[t,T]$. It is interesting to note that we obtain the following
decompositions for $(\Qrob , \FF)$-supermartingales $U^\si(\tau)$ and $\wt U^\si(\tau)$
\begin{gather*}
U^\si(\tau) = \check{M}^\si(\tau)- \check{A}^\si(\tau)- B^\si(\tau),\quad
\wt U^\si(\tau) = \check{M}^\si(\tau)- \check{A}^\si(\tau)- \wt{B}^\si(\tau).
\end{gather*}
Recall that $\check{U}^\si_t(\tau)= U^\si_t(\tau)$. We then proceed
as in the proof of Proposition \ref{propsg21}.
\end{proof}

\subsubsection{No-Arbitrage Bounds in Continuous Time}

The following theorem furnishes the no-arbitrage bounds for a price $\pi^\Aset_t(\Game,s)$ of a combined tranche in the continuous-time framework. As expected, its conclusion is exactly the same as in discrete time setup (see Theorem \ref{thmro10}).

\begin{theorem}\label{thmsg30}
Assume that a continuous time game $\Game$ satisfies Assumption \ref{assrm09}. Let us fix $\Aset\subseteq\Mset$ and $t \in [0,T]$, and let us consider the tranche $\MGCC^\Aset_t(\Game,s)$.
There is no arbitrage in $(B,S,\Aset)$ if and only if the discounted price $\widehat \pi^\Aset_t(\Game,s) = B^{-1}_t\pi^\Aset_t(\Game,s)$ at time $t$ satisfies
\begin{gather}\label{eqro301}
\esssup_{\tau\in\Strat^\Aset_t(s)} \essinf_{\si\in\Strat^{-\Aset}_t(s)} \EMM\Big( \widehat V^\Aset_{\stp (s)}
(\tau,\si)\,\Big|\,\Filt_t\Big)
\leq \widehat\pi^\Aset_t(\Game,s)
\leq \essinf_{\si\in\Strat^{-\Aset}_t(s)} \esssup_{\tau\in\Strat^\Aset_t(s)} \EMM\Big( \widehat V^\Aset_{\stp (s)}
(\tau,\si)\,\Big|\, \Filt_t\Big).
\end{gather}
\end{theorem}

\begin{proof}
In view of Propositions \ref{propsg21} and \ref{propog21}, the arguments used in the proof of Theorem \ref{thmsg30} are exactly the same as the demonstration of Theorem \ref{thmro10} and thus we omit the details.
\end{proof}

\subsection{A Counter-Example} \label{countr}

Unlike for two-person game options, the arbitrage-free condition of Definition \ref{defrm03} restricts the price to an interval, rather than produces a single value. This is due to the fact that we are considering a general game $\Game$ with arbitrary payoff functions. It is difficult to make further progress without being more specific about $\Game$, for instance, by restricting
attention to zero-sum games. It is worth stressing that in the non-zero case, even if the inequalities \eqref{eqro301} are satisfied and arbitrage prices of single tranches are uniquely defined, they may still be problematic from the financial point of view.
The following example illustrates this statement.

\begin{example}\label{eqrx12} {\rm
This example is based on the famous \emph{prisoner's dilemma.} Specifically, let $\Game$ be a deterministic game with two players, each having the action space of $\Strat^i = \{0, 1\}$ for $i=1,2$. The payoff function is given by
\begin{center}
\begin{tabular}{| c | c  c |}
\hline
$V(\tau,\si)$ & $\tau=0$ & $\tau=1$ \\ \hline
$\si=0$ & $(1,1)$ & $(-1,2)$ \\
$\si=1$ & $(2,-1)$ & $(0,0)$ \\
\hline
\end{tabular}
\end{center}
Note that there is an optimal equilibrium at $(\tau^*,\si^*)=(1,1)$ and thus the unique value is $(0,0)$.

Consider the contract $\MGCC(\Game)$ in a market model with null interest rate.
It is easy to see that for each individual tranches, the upper and lower prices are both $\$0$. Hence the arbitrage price for each tranche must be $\pi^1(\Game)=\pi^2(\Game)=0$.  But when we consider the collection of tranches $\Mset=\{1,2\}$, the unique arbitrage price is
\begin{gather}
\pi^\Mset(\Game) = \esssup_{s\in\Strat} \big( V^1(s)+V^2(s)\big) = V^1(1,1) + V^2(1,1)=2.
\end{gather}
In particular, note that $\pi^1(\Game)+\pi^2(\Game)\neq \pi^\Mset(\Game)$.

Now interpret this from the perspective of the issuer and the two prospective holders who are about to enter the contract $\MGCC(\Game)$. For the issuer, the total selling price for the two tranches cannot be less than $\$2$, since there is nothing stopping the two holders from working together for a combined payoff of $\$2$. For the prospective holders of the tranche  1, the buying price cannot be more than $\$0$, since there is nothing stopping the holder of the tranche  2 from playing 1, and limiting the payoff of the tranche  1 to at most 0. In fact, it is strictly better for the holder of the tranche  2 to play 1 than 0 in all cases. A similar analysis can be made for the prospective holder of the tranche~2.

So it is impossible for any deal to strike between the three parties. The issuer will always be demanding more than what the prospective holders are willing to pay. The only way to resolve this situation is for at least two of the three parties to agree to work together and share their payoffs in some way. For example, the two holders agree to buy the two tranches for $\$2$ and play $(0,0)$; or the issuer agrees to sell the tranches for $\$0$, but requiring one of the holders to play 1.}
\end{example}

\subsection{No-Arbitrage Versus Super-Hedging}\label{sec7da4}

We have shown in Theorems \ref{thmro10} and \ref{thmsg30} that there is no arbitrage in the market $(B,S,\Aset)$ if and only if the discounted price of the tranche  $\MGCC^\Aset_t(\Game,s)$ satisfies
\begin{gather}\label{eqrx01}
\underline \pi^\Aset_t(\Game,s)\leq \widehat\pi^\Aset_t(\Game,s) \leq \overline \pi^\Aset_t(\Game,s)
\end{gather}
where the \emph{upper price} $\overline \pi^\Aset_t(\Game,s)$ and the \emph{lower price} $\underline \pi^\Aset_t(\Game,s)$ are defined by
\begin{align*}
\overline \pi^\Aset_t(\Game,s) &:= \essinf_{\si\in\Strat^{-\Aset}_t(s)} \esssup_{\tau\in\Strat^\Aset_t(s)}
\widehat V^\Aset_t (\tau,\si),\\
\underline \pi^\Aset_t(\Game,s) &:= \esssup_{\tau\in\Strat^\Aset_t(s)} \essinf_{\si\in\Strat^{-\Aset}_t(s)}
\widehat V^\Aset_t (\tau,\si),
\end{align*}
where
\[
\widehat V^\Aset_t(\tau,\si ) := \EMM\big( \widehat V^\Aset_{\stp} (\tau,\si)\,\big|\,\Filt_t\big).
\]
In Theorems \ref{thmro10} and \ref{thmsg30}, the upper and lower prices were derived using arbitrage pricing arguments. As expected, they also have an interpretation in terms of super-hedging strategies introduced in Definition \ref{defrm70}.
To this end, we will need the following property.

\begin{lemma}\label{proprx12}
There exist two sequences of strategies, $(\tau_n)_{n\in\N}$ from $\Strat^{\Aset}_t(s)$ and $(\si_n)_{n\in\N}$ from $\Strat^{-\Aset}_t(s)$, such that the sequences
\[
\bigg(\esssup_{\tau\in\Strat^\Aset_t(s)} \widehat V^\Aset_t (\tau,\si_n) \bigg)_{n\in\N} \quad\text{and}\quad \bigg(\essinf_{\si\in\Strat^{-\Aset}_t(s)} \widehat V^\Aset_t (\tau_n,\si ) \bigg)_{n\in\N}
\]
are almost surely non-increasing and non-decreasing, respectively. Furthermore, they converge to the upper and lower prices,
specifically,
\begin{gather*}
\lim_{n\to\infty} \esssup_{\tau\in\Strat^\Aset_t(s)} \widehat V^\Aset_t (\tau,\si_n) = \inf_{n\in\N} \esssup_{\tau\in\Strat^\Aset_t(s)} \widehat V^\Aset_t (\tau,\si_n)=  \overline \pi^\Aset_t(\Game,s),
\end{gather*}
and
\begin{gather*}
\lim_{n\to\infty} \essinf_{\si\in\Strat^{-\Aset}_t(s)} \widehat V^\Aset_t (\tau_n,\si) = \sup_{n\in\N} \essinf_{\si\in\Strat^{-\Aset}_t(s)} \widehat V^\Aset_t (\tau_n,\si)=\underline \pi^\Aset_t(\Game,s).
\end{gather*}
\end{lemma}

\begin{proof}
Using arguments identical to Lemma \ref{proprm31}, the sets
\begin{gather}
\bigg\{ \esssup_{\tau\in\Strat^\Aset_t(s)} \widehat V^\Aset_t (\tau,\si) : \si\in\Strat^{-\Aset}_t(s)\bigg\},\quad
\bigg\{ \essinf_{\si\in\Strat^{-\Aset}_t(s)} \widehat V^\Aset_t (\tau,\si) : \tau\in\Strat^\Aset_t(s)\bigg\}
\end{gather}
have the lattice property. Hence the existence of the desired sequences follows by part (i) in Lemma~\ref{lemrm30}.
\end{proof}

Suppose the issuer holds the tranche $-\Aset$ and plays the strategy $\si\in\Strat^{-\Aset}_t(s)$ while super-hedging using $(\phi^\si,\si)$. The upper price $\overline \pi^\Aset_t(\Game,s)$ is the infimum of $\widehat \Val_t(\phi^\si)$ (the discounted time $t$ value) over all possible $\si\in\Strat^{-\Aset}_t(s)$. Even though the infimum is not necessarily achieved, Lemma \ref{proprx12} shows that there exists a sequence of $(\si_n)_{n\in\N}$ such that $\widehat \Val_t(\phi^{\si_n})$ is non-increasing and converging to $\overline \pi^\Aset_t(\Game,s)$.
Intuitively, the buyer of the tranche  $\Aset$ will be reluctant to pay more than $\overline \pi^\Aset_t(\Game,s)$ at time $t$. Otherwise, his expected discounted payoff could be restricted to be lower than the price paid and, as shown by Theorems \ref{thmro10} and \ref{thmsg30}, leads to an issuer's arbitrage.

Similarly, the holder of the tranche $\Aset$ plays the strategy $\tau\in\Strat^\Aset_t(s)$ and super-hedges the negative of his payoff using $(\phi^\tau,\tau)$. The lower price $\underline \pi^\Aset_t(\Game,s)$ is the supremum of $-\widehat \Val_t(\phi^\tau)$ over all possible $\tau\in\Strat^\Aset_t(s)$. Again the supremum is not necessarily achieved, but Lemma \ref{proprx12} shows that there exists a sequence of $(\tau_n)_{n\in\N}$ such that $-\widehat \Val_t(\phi^{\si_n})$ is non-decreasing and converges to $\underline \pi^\Aset_t(\Game,s)$.
Intuitively, it is undesirable for the issuer to sell the tranche  $\Aset$ for a price cheaper than $\underline \pi^\Aset_t(\Game,s)$. Otherwise the holder of the tranche  $\Aset$ will be able to guarantee an expected discounted payoff greater than the price. As shown in Theorems \ref{thmro10} and \ref{thmsg30}, this leads to a holder's arbitrage.

\section{Pricing of Multi-Person Game Claims} \label{sec7e2}

In the previous section, we focussed on arbitrage valuation of a fixed combined tranche $\Aset $ of $\MGCC^\Aset_t(\Game,s)$. In Theorems \ref{thmro10} and \ref{thmsg30}, we have shown that, in the absence of arbitrage in the market model $(B,S,\Aset)$ (see Definition \ref{defrm03}), an arbitrage price at time $t$ of a combined tranche $\Aset $ lies between the upper and lower prices, that is, \eqref{eqrx01} holds. If there is no arbitrage in $(B,S,\Aset)$ for \emph{any} non-empty subset of tranches $\Aset\subseteq\Mset$, then \eqref{eqrx01} actually describes a total of $2(2^m-1)$ inequalities, which are not easy to handle,
 in general. In this section, we shall further explore how the prices of different collections of tranches are related to one another. In particular, we address in this section the following questions:
\begin{itemize}
\item Apart from \eqref{eqrx01}, should there be other constraints on prices of tranches?
\item Should the price of a combined tranche be equal to the sum of prices of all single tranches, which
constitute a combined tranche?
\item Under which assumptions, a consistent valuation of all tranches (hence of a multi-person game claim)
is feasible?
\end{itemize}

\subsection{Super-Additivity of Lower Prices}\label{sc7e2}

We first show that the lower prices $\underline \pi^\Aset_t(\Game,s)$ are inherently super-additive.

\begin{proposition}\label{propry02}
The lower price of $\MGCC^\Aset_t(\Game,s)$, which is defined by
\[
\underline \pi^\Aset_t(\Game,s) = \esssup_{\tau\in\Strat^\Aset_t(s)} \essinf_{\si\in\Strat^{-\Aset}_t(s)} \widehat V^\Aset_t (\tau,\si),
\]
has the following super-additive property: if $\Aset_1, \Aset_2,\ldots, \Aset_k$ is a partition of $\Aset$, then
\[
\underline \pi^\Aset_t(\Game,s) \geq \sum_{i=1}^k \underline \pi^{\Aset_i}_t(\Game,s).
\]
\end{proposition}

\begin{proof}
According to Lemma \ref{proprx12}, for each of the sets $\Aset_i$, we may choose a sequence $(\tau^{\Aset_i}_n)_{n\in\N}$ from $\Strat^{\Aset_i}_t(s)$ such that the sequence
\[
\essinf_{\si^{-\Aset_i}\in\Strat^{-\Aset_i}_t(s)} \widehat V^{\Aset_i}_t
\left(\tau^{\Aset_i}_n, \si^{-\Aset_i}\right),\quad n\in\N
\]
is non-decreasing and converges to $\underline \pi^{\Aset_i}_t(\Game,s)$.
Now define the strategy $\tau_n\in\Strat^\Aset_t(s)$ for $n\in\N$ by $\tau_n:=(\tau^{\Aset_1}_n,\ldots,\tau^{\Aset_k}_n)$. For each $i=1,\ldots,k$ and for all $\si\in\Strat^{-\Aset}_t(s)$, we have
\[
\widehat V^{\Aset}_t (\tau_n,\si) = \sum_{i=1}^k \widehat V^{\Aset_i}_t (\tau_n,\si) \geq \sum_{i=1}^k \essinf_{\si^{-\Aset_i}\in\Strat^{-\Aset_i}_t(s)} \widehat V^{\Aset_i}_t \left( \tau^{\Aset_i}_n,\si^{-\Aset_i}\right).
\]
Since this holds for all $\si\in\Strat^{-\Aset}_t(s)$, the above inequality implies that
\[
\essinf_{\si\in\Strat^{-\Aset}_t(s)} \widehat V^{\Aset}_t (\tau_n,\si) \geq \sum_{i=1}^k \essinf_{\si^{-\Aset_i}\in\Strat^{-\Aset_i}_t(s)} \widehat V^{\Aset_i}_t \left( \tau^{\Aset_i}_n,\si^{-\Aset_i}\right)
\]
and thus
\begin{align*}
\underline \pi^\Aset_t(\Game,s) &= \esssup_{\tau\in\Strat^\Aset_t(s)} \essinf_{\si\in\Strat^{-\Aset}_t(s)} \widehat V^{\Aset}_t (\tau,\si) \geq \sup_{n\in\N} \essinf_{\si\in\Strat^{-\Aset}_t(s)} \widehat V^{\Aset}_t (\tau_n,\si)  \\
& \geq \sup_{n\in\N}  \sum_{i=1}^k \essinf_{\si^{-\Aset_i}\in\Strat^{-\Aset_i}_t(s)} \widehat V^{\Aset_i}_t
\left(\tau^{\Aset_i}_n, \si^{-\Aset_i}\right)
 = \sum_{i=1}^k \underline \pi^{\Aset_i}_t(\Game,s),
\end{align*}
which completes the proof.
\end{proof}
The super-additive property from Proposition \ref{propry02} can be easily explained as follows. First note that, by Lemma \ref{proprx12}, the holder of any individual tranche may choose a strategy to guarantee an expected discounted payoff which is arbitrarily close to its lower price. So by playing these strategies over a collection of tranches, one can guarantee a total expected discounted payoff arbitrarily close to the sum of individual lower prices, which is in turn no greater than the lower price of the collection of tranches.

Unfortunately, no analogous additive properties exist for the upper prices $\overline \pi^\Aset_t(\Game,s)$. Furthermore, from the analysis so far, there is no reason for the prices $\pi^\Aset_t(\Game,s)$ themselves to satisfy any additive properties either. Does that mean the bounds described by \eqref{eqrx01} are the only important constraints when pricing game contingent claims?

\subsection{Additivity of Prices}\label{sec72}

As illustrated by Example \ref{eqrx12}, even if the bounds of \eqref{eqrx01} hold (and thus the no-arbitrage property of Definition \ref{defrm03} is satisfied), they do not necessarily lead to sensible pricing of the contract. There are two main causes of this discrepancy. First, each party simply cannot rule out the possibilities of the others colluding. This perceived possibility of collusion can also be interpreted as the uncertainty over the identity of the other holders, that is, multiple holders (or issuer) may actually be the same person. Second, the arbitrage prices in Example \ref{eqrx12} do not satisfy a natural requirement of price additivity.

\begin{definition}\label{defrx45}
{\rm We say that the \emph{price additivity} holds whenever}
\begin{gather*}
\pi^\Aset_t(\Game,s)=\sum_{k\in\Aset}\pi^k_t(\Game,s),\quad\forall\,\Aset\subseteq\Mset.
\end{gather*}
\end{definition}

In general, the problem highlighted by Example \ref{eqrx12} occurs under the following conditions. Suppose there are $k$ prospective buyers, who intend to purchase tranches $\Aset_1,\ldots, \Aset_k$, respectively, while the issuer decides to keep the tranche $-\Aset=-(\Aset_1 \cup\cdots\cup \Aset_k)$. It is not possible to reach agreements on prices if the lower price of the tranche $\Aset$ exceeds the sum of upper prices of the tranche $\Aset_i$,
\begin{gather}\label{eqrx50}
\sum_{i=1}^k \overline \pi^{\Aset_i}_t(\Game,s) < \underline \pi^{\Aset}_t(\Game,s) \quad\implies\quad \sum_{i=1}^k \pi^{\Aset_i}_t(\Game,s) < \pi^{\Aset}_t(\Game,s).
\end{gather}
Indeed, the issuer will then always demand more than the what the holders are willing to pay. However, if the price additivity holds, then strict inequality \eqref{eqrx50} cannot occur.

Example \ref{eqrx12} offered some justification to why the additivity of prices is a desirable condition. To further consolidate this belief, we shall introduce a second form of arbitrage called \emph{immediate reselling arbitrage}. In the previous section, Theorems \ref{thmro10} and \ref{thmsg30} prevent enter-and-hold arbitrage opportunities (as per Definition \ref{defrm03}), which are executed by holding a fixed collection of tranches from time $t$ until the settlement time. An immediate reselling arbitrage simply involves buying some tranches at time $t$ and then instantly reselling them in a different configuration with the goal of collecting a higher total price. We will define this formally.

\begin{definition} \label{imm33} {\rm
Consider the game contingent claim $\MGCC_t(\Game,s)$ at time $t$. For each $\Aset\subseteq\Mset$, denote the price of the tranche $\Aset$ by $\pi^\Aset_t(\Game,s)$. A \emph{immediate reselling arbitrage} consists of two different partitions of some subset $\Aset\subseteq\Mset$,
\[
\Aset = \Aset_1 \cup \cdots \cup \Aset_k = \Aset'_1\cup\cdots\cup\Aset'_l,
\]
such that the inequality
\[
\sum_{i=1}^k \pi^{\Aset_i}_t(\Game,s) < \sum_{j=1}^l \pi^{\Aset'_j}_t(\Game,s)
\]
holds on some $\Filt_t$-measurable event $A$ with a positive probability.}
\end{definition}

In European and American option pricing models, immediate reselling arbitrages are implicitly avoided due to the law of one price. In multi-person game contingent claims, the ability to purchase collections of tranches has reintroduced this possibility. However,
it is easy to see that immediate reselling arbitrages are avoided whenever the price additivity holds.
The following result is an immediate consequence of Definition \ref{imm33}.

\begin{proposition}\label{propry20}
There is no immediate reselling arbitrage if and only if the additivity of prices holds. In other words,
\[
\pi^{\Aset}_t(\Game,s)=\sum_{k\in\Aset}\pi^k_t(\Game,s),\quad\forall\,\Aset\subseteq\Mset.
\]
\end{proposition}

By combining Proposition \ref{propry20} with Theorems \ref{thmro10} and \ref{thmsg30}, we obtain the following characterisation of prices which avoids both types of arbitrage mentioned. As usual, we consider a multi-person stochastic game whose payoffs satisfy the appropriate conditions (Assumption \ref{assrm09} for the discrete-time set-up and Assumption \ref{assrm09} for the continuous-time set-up).

\begin{theorem}\label{thmry30}
 For every $\Aset\subseteq\Mset$, let $\widehat \pi^\Aset_t(\Game,s)$ be the discounted time $t$ price of the combined tranche $\Aset $ in $\MGCC^\Aset_t(\Game,s)$. Then there is no arbitrage in $(B,S,\Aset)$ for all $\Aset\subseteq\Mset$ and no immediate reselling arbitrage if and only if the following two conditions are satisfied: \\
(i) the equality
\[
\widehat \pi^\Aset_t(\Game,s)=\sum_{k\in\Aset}\widehat \pi^k_t(\Game,s)
\]
holds for all $\Aset\subseteq\Mset$,
\\ (ii) the random vector $(\widehat \pi^1_t(\Game,s),\ldots,\widehat \pi^m_t(\Game,s) )$ lies in the following random subset of $\R^m$,
\begin{gather}\label{eqry30}
\bigg\{(X_1,\ldots,X_m) : \esssup_{\tau\in\Strat^\Aset_t(s)} \essinf_{\si\in\Strat^{-\Aset}_t(s)}
\widehat V^\Aset_t (\tau,\si) \leq \sum_{i\in\Aset} X_i\\
\leq \essinf_{\si\in\Strat^{-\Aset}_t(s)} \esssup_{\tau\in\Strat^\Aset_t(s)} \widehat V^\Aset_t
(\tau,\si),\quad\forall\,\Aset\subseteq\Mset \bigg\},\nonumber
\end{gather}
where $X_1, \dots , X_m$ are ${\cal F}_t$-measurable random variables.
\end{theorem}

\begin{remark} {\rm
In some game contingent claims $\MGCC(\Game)$, such as Example \ref{eqrx12}, the conditions of Theorem \ref{thmry30} cannot be satisfied because the set in \eqref{eqry30} is empty. This means that some game contingent claims $\MGCC(\Game)$ simply do not admit prices under our arbitrage-free assumptions. Occasionally this may be remedied if the identities of the holders (or the possibility of collusion) is known to everyone in advance, but in these cases $\MGCC(\Game)$ is equivalent to a (simpler) contract with less tranches. Another solution is to alter the definitions of arbitrage and weaken some of the pricing bounds. In any case, it is fair to admit that condition \eqref{eqry30} is not handy and thus we will present more explicit conditions in
Theorem \ref{thmrx10} below.}
\end{remark}

\subsection{Optimal Equilibrium and Arbitrage Prices}

From a purely game-theoretic perspective, the upper and lower prices are also the minimax and maximin values of the coalition $\Aset$ in the subgame of $\Game$ on $[t,T]$. Note that the expected discounted payoff under the martingale measure $\Qrob$ is used in this case. As shown in Section \ref{sec7da4}, the properties of minimax and maximin strategies may be interpreted using super-hedging arguments.

In general, the infimum in the upper price and the supremum in the lower price are not necessarily achieved. But if $\Game$ is known to have an optimal equilibrium $s^*$ then, from Corollary \ref{coraa03}, the maximin and minimax values for each player (or equivalently, the lower and upper prices of individual tranches) are achieved by $\widehat V_t(s^*)$, i.e., by the unique value of the game. Consequently, the vector of prices of individual tranches must coincide with the unique value of the game.

 As demonstrated by means of Example \ref{eqrx12}, the mere existence of optimal equilibria still does not lead to additive prices. However, Proposition \ref{propab01a1} in the appendix shows that if the game $\Game$ is also known to satisfy some further conditions (e.g., zero-sum), then coalition values exist and satisfy the additivity property. In that case, the prices of combined tranches are also unique and the price additivity holds.

We summarise these results in the following theorem. Note that the date $t \in [0,T)$ is fixed, but arbitrary.
We work under our standing assumptions, that is, Assumption \ref{assrm09} for the discrete-time case and Assumption \ref{assrm09} for the continuous-time case.

\begin{theorem} \label{thmrx10}
Consider a game $\Game$  and the associated game contingent claim $\MGCC(\Game)$. Assume that the strategy profile $s_{[0,t)}\in\Strat_{[0,t)}$ is chosen before time $t$ and $\stp(s)\geq t$. Let the expected discounted payoff under the unique martingale measure $\Qrob $ be denoted by $\widehat V_t(\tau,\si) = \EMM\big( \widehat V_{\stp} (\tau,\si)\,\big|\,\Filt_t\big)$.

Suppose that the subgame on the time interval $[t,T]$ has an optimal equilibrium $s^*=(\tau^*,\si^*)\in\Strat_t(s)$, meaning that
\begin{gather}
\widehat V^k_t (\tau^*,\si^*) = \essinf_{\si\in \Strat^{-k}_t(s)} \widehat V^k_t (\tau^*, \si) = \esssup_{\tau\in\Strat^k_t(s)}\widehat V^k_t (\tau ,\si^*),\quad\forall\,k\in\Mset.
\end{gather}
(i) Then, for any $k\in\Mset$, there is no arbitrage in $(B,S,k)$ if and only if the discounted time $t$
price of a single tranche $k$ satisfies
\[
\widehat \pi^k_t(\Game,s)=\widehat V^k_t (s^*).
\]
\\ (ii) Assume, in addition, that the optimal equilibrium $s^*$ satisfies the following condition
\begin{gather}\label{eqrx101}
\sum_{k\in\Mset} \widehat V^k_t (s^*) = \esssup_{s'\in\Strat_t(s)} \sum_{k\in\Mset} \widehat V^k_t (s').
\end{gather}
Then, for any $\Aset\in\Mset$, there is no arbitrage in $(B,S,\Aset)$ if and only if the discounted time $t$
price of a combined tranche $\Aset$ has the following additivity property
\begin{gather}\label{eqrx102}
\widehat \pi^\Aset_t(\Game,s)=\sum_{k\in\Aset}\widehat V^k_t (s^*) = \widehat V^\Aset_t (s^*).
\end{gather}
Furthermore, there is no immediate reselling arbitrage if and only if \eqref{eqrx102} is satisfied for all $\Aset\subseteq\Mset$.
\end{theorem}

\begin{proof}
(i) Corollary \ref{coraa03} from the appendix shows that all optimal equilibria achieve the minimax and maximin values simultaneously. Since we know from Theorems \ref{thmro10} and \ref{thmsg30} that an arbitrage price must lie between the upper and lower prices, which are identical to the minimax and maximin values, an arbitrage price must in fact coincide with the unique value achieved by any optimal equilibrium.

\noindent (ii) Proposition \ref{propab01a1} from the appendix states that, under \eqref{eqrx101}, the coalition values exist and enjoy the additivity property. Formula \eqref{eqrx102} thus follows immediately from Theorems \ref{thmro10} and \ref{thmsg30}, since the arbitrage price of a combined tranche must coincide with the coalition value. The lack of the immediate reselling arbitrage is a consequence of Proposition \ref{propry20} and the price additivity.
\end{proof}

\subsection{Classes of Solvable Multi-Person Games}

Part (ii) in Theorem \ref{thmrx10} is of great significance, since it further justifies the choice of the optimal equilibrium as a solution concept when studying multi-person game contingent claims, as it produces unique arbitrage prices for individual tranches.
In general, as shown in Section \ref{countr}, this would not be possible to achieve with weaker solution concepts, such as the Nash equilibrium.

Let us thus focus on multi-player games that admit an optimal equilibrium under the unique martingale measure $\Qrob$.
In this situation, the existence of unique and additive arbitrage prices for single tranches of a multi-person game claim
hinges on condition \eqref{eqrx101}. Hence the following natural problem arises: describe sufficiently large classes of multi-player games, which admit an optimal equilibrium and satisfy condition \eqref{eqrx101}.

This issue is of interest even in the case of two-person non-zero-sum Dynkin games, although it is fair to say that in most existing papers devoted to Dynkin games, the authors focus on the existence of a Nash equilibrium (see, for instance, Hamad\`{e}ne and Hassani \cite{HH-2011,HH-2012}, Hamad\`{e}ne and Zhang \cite{Hamadene},  Laraki and Solan \cite{Laraki}, Peskir \cite{Peskir} or Ohtsubo \cite{Ohtsubo}). In particular, in \cite{Hamadene}, the authors examine indifference pricing of financial contracts, which are based on non-zero-sum Dynkin games. This should be contrasted with the case of a zero-sum Dynkin game, for which the arbitrage
valuation theory can be applied in a relatively straightforward manner, as demonstrated by Kifer \cite{Kifer1} (for related results in this vein, see also \cite{Dolinsky,Dolinsky11,Kallsen1,Kallsen2,Kifer2,Kuehn,Kyprianou}).

For the multi-player case with $m \geq 3$, an explicit specification of a sensible game admitting an optimal equilibrium is a challenging issue, which was studied in some recent works (see \cite{Guo2,Guo3,Nie}). Due to limited space, we do not discuss this issue in detail here, and we instead restrict ourselves to providing some examples of multi-person games with an optimal equilibrium satisfying condition~\eqref{eqrx101}. In particular, in Section \ref{ssc4.5}, we solve a particular valuation problem for a multi-person contract with puttable tranches, which was described in Section \ref{mutyy}.

As a first example of a game admitting an optimal equilibrium, we may quote any two-player zero-sum game in which only one player has meaningful decisions to make or, equivalently, any single player game. In fact, this class of relatively simple stochastic games covers most of conventional financial derivatives in which the `issuer' has no right to make any decisions about a sold contract.
Obviously, we have in mind here plain-vanilla American options, for which the valuation hinges on solving an optimal stopping problems, but also more complex derivatives of an American style with some `exotic' features, which may allow the holder
to make other relevant decisions, not only exercising (e.g., an American chooser option). Original examples of more sophisticated multi-player games satisfying the assumptions of part (ii) in Theorem  \ref{thmrx10} are collected in the following proposition.

\begin{proposition}
Assume that a game $\Game $ belongs to either of the following classes of multi-player games: \hfill \break
(i)  zero-sum (or constant-sum) games that admits an optimal equilibrium, \hfill \break
(ii) multi-period redistribution games, denoted as $\MRG(X,\alpha)$, \hfill \break
(iii) affine stopping games, denoted as $\ASG(X,G)$, with a $\Kmot$-matrix $G$ with non-negative column sums.\hfill \break
Then ${\cal G}$ has an optimal equilibrium satisfying condition \eqref{eqrx101}
\end{proposition}

\begin{proof}
(i) This part is rather obvious, since both sides in \eqref{eqrx101} vanish identically (or are equal to the same constant)
for a zero-sum (or a constant-sum) game.

\noindent (ii) Multi-period redistribution games $\MRG(X,\alpha)$ are studied in \cite{Guo2}, where it was shown that these
zero-sum (non-zero-sum, resp.) redistribution games have an optimal equilibrium 
by virtue of Theorem 2 (Theorem 3, resp.).

\noindent (iii)
Recall that $\Kmot$-matrix is a square matrix with a non-negative determinant, positive proper principal minors, and with  non-positive off-diagonal terms (see Definition 3.2 in \cite{Guo3}). The desired properties of an affine stopping game $\ASG(X,G)$ with a $\Kmot$-matrix $G$ were established in Theorem 3.2 in \cite{Guo3}.
\end{proof}

In the same way a general game contingent claim is defined, one can define a multi-person game contingent claim based on a
multi-player game $\Game $.  Of course, a classic example of this approach is the definition of the two-person game claim, which corresponds to a zero-sum two-person Dynkin game. It should be mentioned, however, that in a multi-period (or continuous-time) setup, the redistribution and affine games introduced in Guo \cite{GuoPhD}  and Guo and Rutkowski \cite{Guo1,Guo2,Guo3} are defined recursively, meaning that the payoffs at settlement depend on the value of analogous virtual game continued after the settlement date (this concept is reminiscent of a continuation value for an American or game claim). For this reason, the optimal equilibria are obtained either through the backward induction or by solving a multi-dimensional reflected BSDE (see Nie and Rutkowski \cite{Nie}). The latter approach hinges on a multi-dimensional extension of the one-dimensional doubly-reflected BSDE for the Dynkin game studied in Cvitani\'c and Karatzas \cite{Cvitanic}.

\subsection{Pricing of Multi-Person Contracts with Puttable Tranches} \label{ssc4.5}

We conclude this work by examining arbitrage pricing of a multi-person contract with puttable tranches. Recall that the mechanism of this contract was outlined in Section \ref{mutyy}. The key idea is to reduce the valuation problem into the search of optimal equilibria for the associated multi-player stochastic game under a unique martingale measure for the underlying market model. We then apply the results from our previous papers \cite{Guo1,Guo2} to find the solution to the game. In fact, the results from \cite{Guo2} need to be slightly extended, since the redistribution games studied in \cite{Guo2} are formally postulated to terminate as soon as at least one of the players decides to exercise (i.e., stop the game), whereas in Section \ref{mutyy} we introduced a contract, which always continues till its maturity date $T = T_{n+1}$. Note that we work here with either a discrete- or continuous-time setting.

Let us consider the date $T_n$. Suppose that there exists a unique martingale measure $\Q$ for the underlying market model, say, the CRR binomial model or the Black and Scholes model. We assume, for simplicity, that the interest rate is null (otherwise, one would work with the discounted values). Then the prices $P^1_n,\ldots, P^m_n$ of tranches at time $T_n$, before decisions of all parties at time $T_n$ were made (dubbed the {\it continuation values} at time $T_n$), are prices of European or American options and thus they can be found by the usual arbitrage pricing method independently for each tranche. For instance, if the terminal payoffs of European contingent claims are $X^1_{n+1}, \dots , X^m_{n+1}$, then we obtain the equality $P^i_n =\EMM( X^{i}_{n+1}\,|\,\Filt_{T_n})$ for every $i$ where $\mathbb{F} = (\Filt_{t})_{t \in [0,T]}$ is the underlying filtration.

A strategy for the holder of tranche $i$ on $[0,T_n]$ is denoted by $s^i$. In the present setting, it
consists of choosing between `hold' or `put' at each date $T_1,\ldots,T_n$. These decisions are made
 by the current holder of tranche $i$ as time progresses, which may not be the same person, due to the nature of the contract.
Let us denote a generic strategy profile by $s$ and its restriction to $[T_l,T_n]$ by $s_l$.
It will be convenient to represent a strategy profile $s_l$ as $s_l = (\tau_l ,\sigma_l )$ where $\tau_l = s^i_l$ ($\sigma_l = s_l^{-i}$, resp.) is the strategy of player $i$ (the strategy profile for all other players, resp.) on $[T_l,T_n]$.
We denote by $\Eset_l (s_l)$ the set of players who decide to put the tranche at time $T_l$ when a strategy profile
$s_l$ is played on $[T_l,T_n]$. To proceed further, we need to introduce some notation, which is adapted from \cite{Guo2}. For any $X = (X^1, \dots, X^m) \in \R^m$ and an arbitrary subset $\Eset\subseteq\Mset$, we define the hyperplane
\[
\H (\Eset , X,c) := \bigg\{ {x} \in \R^m :\, \sum\limits_{i=1}^m x_i = c   \ \mbox{and}\ x_i = X^{i}, \ \forall\, i\in \Eset\bigg\}
\]
and the simplex
\[
\S (X,c):=\bigg\{ {x} \in \R^m :\, \sum\limits_{i=1}^m x_i = c   \ \mbox{and}\
 x_i\geq X^{i}, \, \forall \, i \in {\cal M} \bigg\}.
\]
Let an $\Filt_{T_l}$-measurable random variable $(C^1_l, \dots, C^m_l)$ and an $\Filt_{T_l}$-measurable subset $\Eset_l$ of $\Mset $ be given. Then we define the random, $\Filt_{T_l}$-measurable hyperplane $\H (\Eset_l, X_l,c_l)$ and the random $\Filt_{T_l}$-measurable simplex $\S (X_l,c_l)$, where $c_l := \sum_{i=1}^m C^i_l$. Let us finally observe that, since for the uniform redistribution we have that $\alpha_i = \frac{1}{m}$ for all $i$, the inner product given by equation (3.9) in \cite{Guo2} reduces to the Euclidean inner product and thus the projection $\pi^0$ utilised in Theorem 3.12 in \cite{Guo2} is the usual orthogonal projection in $\R^m$; it will be henceforth denoted as $\pi$.
The next definition hinges on Lemma 3.11 in \cite{Guo2}, which furnishes an explicit characterisation
\eqref{sftvf2} of the effective (or {\it modified}) payoff in a zero-sum single-period redistribution game.
We denote by $V^*_{t+1}$ the {\it value} of the subgame started at $T_{t+1}$ and thus the definition
of the game is recursive (its correctness is formally supported by Proposition \ref{pvrpry20} where the existence and uniqueness
of the value process is established).

\begin{definition} \label{deftuu} {\rm
Assume that an $\mathbb{F}$-adapted process $X_l = (X^1_l, \dots , X^m_l),\, l=1, \dots ,n$ and the continuation value
$P_n = (P^1_n, \dots , P^m_n)$ at time $T_n$ are given. Let $V^i_l(s_l)$ denote the $\Filt_{T_l}$-conditional expected risk-neutral payoff at time $T_l$ for the current holder of tranche $i$ given that a strategy profile $s_l$ is played on $[T_l,T_n]$. In the {\it multi-player game with puttable tranches}, for any strategy profile $s$,  $V_l(s_l)$ is defined recursively through the following expression, for every $l=1, \dots , n$ and $i=1,\ldots,m$,
\begin{gather} \label{sftvf2}
 V^i_l(s_l) := X^i_l \I_{\{\vartheta^i  (s_l) =l\}} + \bigg(\big[\pi_{\H (\Eset_l (s_l), X_l, p_l)}(P_l)\big]^i-P^i_l + \EMM \big(  V^i_{l+1}(s_{l+1})\,|\,\Filt_{T_l} \big)\bigg)\I_{\{\vartheta^i (s_l) >l\}}
\end{gather}
where the equality $\vartheta^i (s_l)= l$ (the inequality $\vartheta^i (s_l) >l$, resp.) means that the holder
of the $i$th tranche decides to put (decides not to put, resp.) at time $T_l$ and
\begin{gather} \label{sfitvf}
p_l :=\sum_{i=1}^m P^i_l
\end{gather}
where $P^i_l :=  \EMM \big( V^{i*}_{l+1} \,|\,\Filt_{T_l} \big)$ for every $i=1,\dots, m$ and $l=1,\dots ,n-1$.
Furthermore, we set $V_0(s) = \EMM( V_{1}(s_1))$ since putting a tranche at time 0 is not allowed.}
\end{definition}

Observe that the {\it continuation value} of tranche $i$ at time $T_l$ is given as $P^i_l =\EMM(V^{i*}_{l+1}\,|\,\Filt_{T_l})$
for every~$i$. In equation \eqref{sftvf2}, the term $X^i_l$ represents the payoff received by holder of tranche $i$ when he decides to put at time $T_l$, and the quantity in front of the indicator $\I_{\{\vartheta^i >l\}}$ represents the expected payoff of holder when he decides not to put the tranche.  Upon summation, we can represent the right-hand side in \eqref{sftvf2} as follows
\begin{gather} \label{sftvf}
 V^i_l(s_l) :=  \EMM \bigg( X^i_{\vartheta^i (s_l)} + \sum_{t=l}^{\vartheta^i (s_l)-1} \Big[\pi_{\H (\Eset_t (s_t), X_t, p_t)} \big(\EMM \big(V^*_{t+1}\,|\,\Filt_{T_t} \big)  \big)-\EMM \big(V^*_{t+1}\,|\,\Filt_{T_t} \big) \Big]^i\,\bigg|\,\Filt_{T_l} \bigg)
\end{gather}
where $\vartheta^i (s_l)$ represents the index such that the random time $T_{\vartheta^i (s_l)} \geq T_l$ is the first instance at which tranche $i$ is put. Hence the summation term in \eqref{sftvf} represents the accumulated redistribution holder $i$ is required to pay as a result of other holders putting their tranches during $[T_l,T_{\vartheta^i (s_l)})$.

It should be noted that game described in Definition  \ref{deftuu} is nothing more than a suitable variant of the game given by Definition 4.1 in \cite{Guo2}. The only difference is that, in \cite{Guo2}, the multi-person game is stopped as soon as at least one player decides to put his tranche. In that case, the expected payoff is determined using, in particular, the value process of a virtual game played from this moment onwards. Therefore, Definition  4.1 in \cite{Guo2} refers in fact to a family of the embedded games that start at any date $T_l$; each of them ends as soon as someone `exercises' (i.e., a holder puts his tranche, using the present terminology).  By contrast, in Definition \ref{deftuu}, we formally specify a single game, which is always played till expiration date $T$, but is otherwise specified in the same way as the game studied in \cite{Guo2}. It is thus natural to conjecture that the two games share the same value process and an optimal equilibrium. The validity of this conjecture will be established in Proposition \ref{pvrpry20}. The proof of this result utilises Theorem 1 in \cite{Guo1} (see also Theorem 3.12 in \cite{Guo2}).

We acknowledge that the game introduced in  Definition  \ref{deftuu} departs from the usual way
in which multi-person extensions of classic Dynkin games are defined (see, for instance, Hamad\`{e}ne and Hassani \cite{HH-2011,HH-2012} or Karatzas and Li \cite{KaratzasLi}),
since we postulate that the expected payoff at time $T_l$ depends not only on exogenously specified payoff
processes $X^1, \dots , X^m$, but also on the value of the embedded game started at $T_l$. This way of defining a game
could seem unusual in the context of the game theory, but it is perfectly consistent with the arbitrage pricing
of derivative securities and the concept of the market value, where the future decisions of investors need to be taken into account. Also, in the case of a two-person game option, the corresponding Dynkin game is formally defined using only
the payoff processes, say $X^1$ and $X^2$. However, an equivalent formulation that utilises processes $X^1 \geq X^2$
and the continuation value $P$ of the contract  (or, equivalently, the associated Dynkin game) is also possible
and in fact it provides a better insight into competitive character of a game option. However, since the gain of one party matches the loss of the counterparty, the trivial mechanism of redistribution is already implicit in the two-person
version of the game. Of course, when $m=2$, the game of Definition \ref{deftuu} reduces to the classic Dynkin game and
a multi-person contract with puttable tranches becomes a two-person game option.

Since it is possible for the issuer to simultaneously hold multiple tranches at one point, the price additivity
(see Definition \ref{defrx45}) becomes an important issue. We hence need to check a sub-zero-sum condition on the value $V^*_0$. The most convenient solution is to enforce a restriction such as \eqref{rfvfv}, so that it will never be optimal for all holders to simultaneously put, and thus the contract will be a zero-sum game between the holders.

To ensure this property in the case of a contract described in Section \ref{mutyy}, one should choose strikes $K_1, \dots , K_m$, classes of options (calls or puts), and the {\it put payoffs} $X^1_l, \dots , X^m_l,\, l=1, \dots ,n$ in such a way that it will be not `optimal' (in the sense of maximisation of the value of a tranche) for all holders to simultaneously give their tranches back to the issuer at some date $T_l$.

\begin{proposition} \label{pvrpry20}
Assume that the zero-sum condition
\begin{gather} \label{rfvfv}
\sum_{i=1}^m X^i_l \leq  \sum_{i=1}^m P^i_l
\end{gather}
is satisfied recursively for $l=1, \dots , n$. Then the multi-player game with puttable tranches admits an optimal equilibrium $s^*$ and a unique value process $V^*=V(s^*)$.
\end{proposition}

\begin{proof}
In the first step, we consider the penultimate date $T_n$. Then the continuation values $P^1_n ,\dots , P^m_n$ are
obtained by computing the continuation values for all tranches at time $T_n$ using the standard arbitrage pricing method.
Consequently, the multi-person game with puttable tranches can be reduced at time $T_n$ to a single-period zero-sum redistribution game $\ZRG(X_n,P_n, \alpha)$ with $\alpha_i = \frac{1}{m}$ for all $i$ (see Definitions 3.1 and 3.10 in \cite{Guo2}). Therefore, the existence of an optimal equilibrium $s^*_n = (\tau^*_n,\sigma^*_n)$ (hence also of a unique value $V^*_n$ for the game at time $T_n$) follows from Theorem 1 in \cite{Guo1} (see also Theorem 3.12 in \cite{Guo2}). Specifically, the unique value for each individual tranche at time $T_n$ equals
\begin{gather*}
V^*_n =V_1(s^*_n)=\pi_{\S (X_n,p_n)}(P_n)
\end{gather*}
provided that
\begin{gather*}
\sum_{i=1}^m X^i_n \leq p_n = \sum_{i=1}^m P^i_n,
\end{gather*}
that is, it is not optimal for all players to put tranches at time $T_n$.

We now proceed by backward induction. We fix $l=1, \dots ,n-1$, and we suppose that the value at time $T_{l+1}$ is $V^*_{l+1}=V_{l+1}(s^*_{l+1})$ where a strategy profile $s^*_{l+1}$ is an optimal equilibrium for the game on the time interval $[T_{l+1},T_n]$, meaning the following inequalities hold for every $i$ and all strategies $\tau_{l+1}$ and $\sigma_{l+1}$
\begin{gather} \label{vrttx}
V^i_{l+1}(\tau^*_{l+1},\si_{l+1}) \geq V^{*i}_{l+1} \geq V^i_{l+1}(\tau_{l+1},\si^*_{l+1}).
\end{gather}
Recall that the vector of continuation values at time $T_l$ is given by $P^i_l :=\EMM(V^{*i}_{l+1}\,|\,\Filt_{T_l})$ for all $i$.
Let us then suppose that the value $V^*_{l+1}$ of the game at time $T_{l+1}$, and thus also the continuation values $P^1_l,\ldots, P^m_l$ of tranches, before players' decisions at time $T_l$ were made, are given.

 We define the strategy profile $s^*_{l}$ on $[T_l,T_n]$ by postulating that: (i)  the set of players who decide to put the tranche at the time $T_l$ equals
\begin{gather} \label{kiki}
\Eset_l(s^*_{l}) = \big\{\, i \in \{ 1,\ldots,m \} : \big[\pi_{\S (X_l,p_l)}(P_l)\big]^i = X^i_l \big\}
\end{gather}
and (ii) $s^*_{l}$ restricted to $[T_{l+1},T_n]$ coincides with $s^*_{l+1}$. Using \eqref{sftvf}--\eqref{sfitvf} and property (ii), we obtain
\begin{gather*}
 V_l(s^*_l) = \pi_{\H (\Eset_l (s^*_l), X_l,p_l)}( P_l )-P_l+\EMM \big( V_{l+1}(s^*_{l+1})\,|\,\Filt_{T_l} \big)
 = \pi_{\H ( \Eset_l (s^*_l),X_l,p_l)}(P_l)
\end{gather*}
where the set $\Eset_l (s^*_l)$ is given by \eqref{kiki} and $p_l$ is given by \eqref{sfitvf}.
We aim to demonstrate that the strategy profile $s^*_{l}$ is an optimal equilibrium on $[T_l,T_n]$ and $V^*_l :=V_l(s^*_{l})$ is the game value at time $T_l$. To this end, we need to show that, for every $i$,
\begin{gather} \label{tutut}
 V^i_{l}(\tau^*_{l},\si_{l})  \geq V^i_l(s^*_{l}) = V^i_l (\tau^*_l,\sigma^*_l) \geq V^i_{l}(\tau_{l},\si^*_{l})
\end{gather}
 for arbitrage strategy profiles $\tau_{l}$ and $\sigma_{l}$ on $[T_l,T_n]$, not necessarily coinciding with $\tau^*_{l+1}$
and $\sigma^*_{l+1}$ on $[T_{l+1},T_n]$.

 From Theorem 1 in \cite{Guo1} (or Theorem 3.12 in \cite{Guo2}), it is known that condition (i) characterises an optimal equilibrium in the single-period game $\ZRG(X_l,P_l, \alpha )$ and thus
\begin{gather} \label{drfe}
 V_l(s^*_l) = \pi_{\H (\Eset_l (s^*_l), X_l,p_l)}(P_l) = \pi_{\S (X_l,p_l)}(P_l)
\end{gather}
provided that $\sum_{i=1}^m X^i_l \leq p_l$. To be more specific, for every $i=1,\dots ,m$ and all strategies $\tau_{l}$ and $\sigma_{l}$ such that $\tau_{l+1} = \tau^*_{l+1}$ and $\sigma_{l+1} = \sigma^*_{l+1}$, we have
\begin{align}  \label{drfe1}
\bar{V}^i_l(\tau^*_{l},\si_{l}) &= \big[ \pi_{\H (\Eset_l(\tau^*_{l},\si_{l}), X_l,p_l)}(P_l) \big]^i \geq
 [ \pi_{\S (X_l,p_l)}(P_l)]^i \\ & \geq \big[ \pi_{\H (\Eset_l(\tau_{l},\si^*_{l}),X_l,p_l)}(P_l) \big]^i = \bar{V}^i_l(\tau_{l}, \si^*_{l}) \nonumber
\end{align}
where we write $\bar{V}^i_l$, rather than $V^i_l$, to emphasise that we deal here with a single-period game on $[T_l,T_{l+1}]$, which is obtained by postulating that all strategy profiles that appear in the formula above 
necessarily coincide with $s^*_{l+1}$ on $[T_{l+1},T_n]$.

To show the Nash equilibrium property of $s^*_l$ at time $T_l$ for the multi-player game given by Definition \ref{deftuu},
let us take any strategy profile $\tau_l$ for the holder of tranche $i$. If $\tau_l$ implies putting tranche $i$ at time $T_l$, then from \eqref{drfe1}, we obtain
\begin{align*}
V^i_l(\tau_{l},\si^*_{l}) &= X^i_l = \big[ \pi_{\H (\Eset_l(\tau_{l},\si^*_{l}), X_l,p_l)}(P_l)\big]^i\leq [\pi_{\S (X_l,p_l)}(P_l)]^i=V^i_l(s^*_{l}).
\end{align*}
If $\tau_l$ does not imply putting tranche $i$ at time $T_l$, then from \eqref{vrttx}, \eqref{drfe} and \eqref{drfe1}, we obtain
\begin{align*}
V^i_l(\tau_{l},\si^*_{l}) &= \big[ \pi_{\H (\Eset_l(\tau_{l},\si^*_{l}), X_l,p_l)}(P_l)\big]^i- P_l^i + \EMM\big(V^i_{l+1}(\tau_{l+1},\si^*_{l+1})\,|\,\Filt_{T_l}\big)\\
&\leq [\pi_{\S (X_l,p_l)}(P_l)]^i- P^i_l + \EMM(V^{*i}_{l+1}\,|\,\Filt_{T_l})=V^i_l(s^*_{l}) .
\end{align*}
To verify the guaranteed payoff inequality, we take an arbitrary strategy profile $\si_l$ and we use similar arguments as above. If $\tau^*_l$ implies putting tranche $i$ at time $T_l$, then from \eqref{drfe1}, we obtain
\begin{align*}
V^i_l(\tau^*_{l},\si_{l}) &= X^i_l = \big[ \pi_{\H (\Eset_l(\tau^*_{l},\si_{l}), X_l,p_l)}(P_l)\big]^i\geq [\pi_{\S (X_l,p_l)}(P_l)]^i=V^i_l(s^*_{l}).
\end{align*}
If $\tau^*_l$ does not imply putting tranche $i$ at time $T_l$, then from \eqref{vrttx}, \eqref{drfe} and \eqref{drfe1}, we obtain
\begin{align*}
V^i_l(\tau^*_{l},\si_{l}) &= \big[ \pi_{\H (\Eset_l(\tau^*_{l},\si_{l}), X_l,p_l)}(P_l)\big]^i- P_l^i + \EMM\big(V^i_{l+1}(\tau^*_{l+1},\si_{l+1})\,|\,\Filt_{T_l}\big)\\
&\geq [\pi_{\S (X_l,p_l)}(P_l)]^i- P^i_l + \EMM(V^{*i}_{l+1}\,|\,\Filt_{T_l})=V^i_l(s^*_{l}) .
\end{align*}
We have thus shown that  \eqref{tutut} holds for every $i$ and all strategies $\tau_{l}$ and $\sigma_{l}$. This means that $s^*_l$ is indeed an optimal equilibrium on $[T_l,T_n]$ and $V^*_l$ is the value of the game at time $T_l$. Finally, the value of the game at time 0 equals $V^*_0=\EMM (V^*_1)$, since putting a tranche at time 0 is not allowed.
\end{proof}

One could also consider a simplified version of the game of Definition \ref{deftuu}, which is specified as follows: we assume that the vector $P_n = (P^1_n, \dots , P^m_n)$ is given and, for any strategy profile $s$ on $[0,T_n]$, we set
\begin{gather} \label{yy778}
V_n(s_n) := \pi_{\H (\Eset_n (s_n), X_n, p_n)} (P_n)
\end{gather}
and we define $V_l(s_l)$ recursively, for every $l=1, \dots , n-1$,
\begin{gather*}
V_l(s_l) := \pi_{\H (\Eset_l (s_l), X_l, p_l)}\big( \EMM \big( V^*_{l+1} \,|\,\Filt_{T_l} \big) \big)
\end{gather*}
where  $V^*_{l+1}$ is the value of the subgame started at $T_{l+1}$ and $p_l,\, l=1, \dots ,n-1$ is given by \eqref{sfitvf}
with $P^i_l :=  \EMM \big( V^{*i}_{l+1} \,|\,\Filt_{T_l} \big)$. For this game, Proposition \ref{pvrpry20} is still valid
and its proof can be essentially simplified.

It is less clear whether the following definition of the game would be suitable:
we assume that the vector $P_n = (P^1_n, \dots , P^m_n)$ is given and,
for any strategy profile $s$ on $[0,T_n]$, we postulate that \eqref{yy778} holds
and  we set, for every $l=1, \dots , n-1$,
\begin{gather*}
V_l(s_l) := \pi_{\H (\Eset_l (s_l), X_l, \bar p_l (s_{l+1}))}\big( \EMM \big( V_{l+1}(s_{l+1})\,|\,\Filt_{T_l} \big) \big)
\end{gather*}
where, for every $l=1, \dots , n-1$,
\begin{gather*}
\bar p_l (s_{l+1}) := \sum_{i=1}^m \EMM \big( V^i_{l+1}(s_{l+1})\,|\,\Filt_{T_l} \big).
\end{gather*}
Unfortunately, this more flexible game is unlikely to possess an optimal equilibrium, in general, and thus its specification should be complemented by further conditions.

\begin{remark} {\rm
One may wonder whether the multi-person game introduced in  Definition  \ref{deftuu} can be solved using a reflected BSDE, as was demonstrated by Cvitani\'c and Karatzas \cite{Cvitanic} for two-person Dynkin games. This is indeed the case, but it requires a judicious specification of the behaviour of a solution to a multi-dimensional BSDE at the boundary.  In \cite{Guo3}, it was shown that the BSDE approach can be applied in a discrete-time setting to solve the so-called {\it multi-player affine games}, which cover the game of Definition  \ref{deftuu} as a special case.  The interested reader is also referred to \cite{Nie} where continuous-time multi-player stopping games are solved using multi-dimensional BSDEs with oblique reflection.}
\end{remark}

\vskip 15 pt


\noindent {\bf Acknowledgement.}
The research of Ivan Guo and Marek Rutkowski was supported under Australian Research
Council's Discovery Projects funding scheme (DP120100895).


\newpage

\section{Appendix: Nash and Optimal Equilibria}\label{sec1d}

We summarise here the basic results for Nash and optimal equilibria in a multi-player game.
Let ${\cal M} = \{1, \dots, m\}$ be the set of players and let $\Strat $ stand for the class of all strategy profiles $s = (s^1, \dots , s^m)$. For each $s \in \Strat $, we denote by $V^k(s)= V^k(s^1, \dots , s^m)$ a (possibly random) payoff of the $k$th player. It is convenient to write  $s = (\si^{k}, \si^{-k})$ and $\Mset^{-k} = \Mset \setminus \{k\}$.  An analogous notational convention will be later applied to any proper subset $\Aset $ of $\Mset $.

\begin{definition} \label{defaa01} {\rm
A strategy profile $\si= (\si^{1},\ldots,\si^{m}) \in\Strat$ is called a \emph{Nash equilibrium}
if the inequality $V^k(\si^{k}, \si^{-k}) \geq V^k(s^k, \si^{-k})$ holds for each $k$ and all $s^k\in\Strat^k$.
In other words, for each $k\in\Mset$,
\begin{gather}\label{eqaa011}
V^k(\si^{k}, \si^{-k}) = \esssup_{s^k\in\Strat^k} V^k(s^k, \si^{-k}).
\end{gather}
Often \eqref{eqaa011} is written as $V^k(\si^{k}, \si^{-k}) \geq V^k(s^k, \si^{-k})$ for all $s^k\in\Strat^k$.}
\end{definition}

Let us now recall the notion of an optimal equilibrium, which is stronger than a by far more widely
used concept of a Nash equilibrium (although the two concepts coincide in the case of a two-person zero-sum game).

\begin{definition} \label{defaa02} {\rm
A Nash equilibrium $\si= ( \si^{1},\ldots,\si^{m})\in\Strat$ is called an \emph{optimal equilibrium} if the inequality
$V^k(\si^k, \si^{-k})\leq V^k(\si^k, s^{-k})$ holds for each $k\in\Mset$ and all $s^{-k} \in\Strat^{-k}$ or,
equivalently, for each $k\in\Mset$,
\begin{gather}\label{eqaa021}
V^k(\si^k, \si^{-k}) = \essinf_{s^{-k}\in \Strat^{-k}} V^k(\si^k, s^{-k}).
\end{gather}
Combining with condition \eqref{eqaa011} of a Nash equilibrium, $\si$ satisfies
\begin{gather}\label{eqaa0221}
V^k(\si^k, \si^{-k}) = \essinf_{s^{-k}\in \Strat^{-k}} V^k(\si^k, s^{-k}) = \esssup_{s^k\in\Strat^k} V^k(s^k, \si^{-k}),
\end{gather}
or equivalently,
\begin{gather}\label{eqaa022}
V^k(\si^k, s^{-k})\geq V^k(\si^k, \si^{-k})
\geq V^k(s^k, \si^{-k}), \quad \forall \, s^k\in\Strat^k, s^{-k}\in \Strat^{-k},
\end{gather}
for each $k\in\Mset$.}
\end{definition}

It is clear that an optimal equilibrium is essentially a saddle point. In addition to the properties of a Nash equilibrium, each player can guarantee his optimal equilibrium payoff without knowing the actions of other players. Furthermore, as shown in Corollary \ref{coraa03} below, all optimal equilibria achieve the same value. The concept of an optimal equilibrium is also closely related to the maximin and minimax values of the game. The \emph{maximin  value} is the maximum payoff player $k$ can guarantee.

\begin{definition}\label{defminimax}{\rm
The \emph{maximin  value} $\underline V^k$ equals
\begin{gather}\label{eqaa51}
	\underline V^k:=\esssup_{s^k\in\Strat^k}\essinf_{s^{-k}\in\Strat^{-k}} V^k(s^k, s^{-k}).
\end{gather}
A \emph{maximin strategy} for player $k$ is any strategy $\wh s^k\in\Strat^k$ such that $\essinf_{s^{-k}\in\Strat^{-k}} V^k(\wh s^k, s^{-k}) = \underline V^k$, that is, a strategy $\wh s^k$ realises the supremum in \eqref{eqaa51}.}
\end{definition}

The \emph{minimax  value} is the lowest payoff that the other players can force upon player $k$.

\begin{definition}\label{deinimax}{\rm
The \emph{minimax value} $\overline V^k$ equals
\begin{gather}\label{eqaa52}
\overline V^k:=\essinf_{s^{-k}\in\Strat^{-k}}\esssup_{s^k\in\Strat^k} V^k(s^k, s^{-k}).
\end{gather}
A \emph{minimax strategy profile} for the player set $\Mset^{-k}$ is any strategy profile $\wh s^{-k}\in\Strat^{-k}$ such that
the equality $\esssup_{s^k\in\Strat^k} V^k(s^k,\wh s^{-k})=\overline V^k$ holds, that is, $\wh s^{-k}$ realises the infimum in \eqref{eqaa52}.}
\end{definition}

\subsection{Value of a Game}

In general, the maximin value never higher than the minimax value, since the players from the set $\Mset^{-k}$ cannot force the payoff of player $k$ to be lower than an amount that can be guaranteed by him. For this reason, the maximin value
 (minimax value, resp.)  is also known as the \emph{lower value} (the \emph{upper value}, resp.).
 The most fundamental properties of an optimal equilibrium are summarised in the following proposition.

\begin{proposition}\label{propaa02}
(i) The inequality $\overline V^k \geq \underline V^k$ is valid for all $k$.\\
(ii) If $\sigma$ is a Nash equilibrium, then $V^k(\sigma)\geq \overline V^k$  for all $k$.\\
(iii) If $\sigma$ satisfies \eqref{eqaa021}, then for all $k$, $V^k(\sigma)\leq \underline V^k$.\\
(iv) If $\sigma$ is an optimal equilibrium, then  $V^k(\sigma)= \overline V^k=\underline V^k$ for all $k$.\\
(v) If $\sigma$ is an optimal equilibrium, then for all $k$ the strategy $\sigma^k$ (the strategy profile $\sigma^{-k}$, resp.)
 is a maximin strategy for player $k$ (a minimax strategy profile the players set $\Mset^{-k}$, resp.).
\end{proposition}

\begin{proof}
(i) For every $\widehat s^k \in \Strat^k$ and $\widehat s^{-k}\in \Strat^{-k} $, we have that
\[
G(\widehat s^{-k} ) := \sup_{s^{k}\in \Strat^k} V^k ( s^k , \widehat s^{-k}) \geq V^k( \widehat s^k  ,\widehat s^{-k} )\geq
\inf_{s^{-k}\in \Strat^{-k} } V^k( \widehat s^k , s^{-k} ) =: H(\widehat s^{k} )
\]
and thus $G(s^{-k} ) \geq H(s^{k} )$ for every $s^{k}$ and $s^{-k}$. Consequently,
\begin{gather*}
\overline V^k = \inf_{s^{-k}\in\Strat^{-k}}\sup_{s^k\in\Strat^k} V^k(s^k ,s^{-k} ) = \inf_{s^{-k}\in\Strat^{-k}} G(s^{-k} )\\
\geq \sup_{s^k\in\Strat^k} H(s^{k} ) = \sup_{s^k\in\Strat^k}\inf_{s^{-k}\in\Strat^{-k}} V^k(s^k ,s^{-k} )=\underline V^k.
\end{gather*}

\noindent (ii) If condition \eqref{eqaa011} holds then
\[
V^k(\si^k,\si^{-k})= \sup_{s^k\in\Strat^k} V^k(s^k,\si^{-k}) \geq \inf_{s^{-k}\in\Strat^{-k}}\sup_{s^k\in\Strat^k} V^k(s^k,s^{-k})=\overline V^k.
\]

\noindent (iii) If condition \eqref{eqaa021} holds then
\[
V^k(\si^k,\si^{-k})= \inf_{s^{-k}\in\Strat^{-k}}V^k(\si^k,s^{-k}) \leq \sup_{s^k\in\Strat^k}\inf_{s^{-k}\in\Strat^{-k}} V^k(s^k ,s^{-k} )=\underline V^k.
\]
\\(iv) If both conditions \eqref{eqaa011} and \eqref{eqaa021} hold, then by (ii) and (iii),
\[
\underline V^k \geq V^k(\sigma)\geq \overline V^k.
\]
In view of (i), we thus obtain the equality $V^k(\sigma)= \overline V^k=\underline V^k$.

\noindent (v) Combining (iv) with condition \eqref{eqaa022}, we have,
\begin{align*}
V^k(\si^k, s^{-k})\geq V^k(\si^k, \si^{-k}) &= \underline V^k, \quad \forall \, s^{-k}\in \Strat^{-k},\\
V^k(s^k, \si^{-k})\leq V^k(\si^k, \si^{-k}) &= \overline V^k, \quad \forall \, s^k\in\Strat^k.
\end{align*}
Hence $\sigma^k$ and $\sigma^{-k}$ are maximin and minimax strategies, respectively.
\end{proof}

Proposition \ref{propaa02} motivates the following definition of values of the game.

\begin{definition}\label{defaa03}{\rm
If the equality $\overline V^k = \underline V^k$ holds, then $V^{*k}:=\overline V^k = \underline V^k$ is the \emph{value for player} $k$. The \emph{value of the game} $\bm{V}^*$ is the vector $(V^{*1},\ldots,V^{*m})$ of values for all players.}
\end{definition}

Since the equality is not necessarily achieved in part (i) in Proposition \ref{propaa02}, the existence of the value is not guaranteed. However, by part (iv) in this proposition, the existence of an optimal equilibrium implies the existence of the
value of the game, so that we may state the following corollary.

\begin{corollary}\label{coraa03}
(i) If the value of the game $\bm{V}^*$ exists, then it is unique.
\\(ii) If there exists an optimal equilibrium $\si$, then the value exists for every player and
\[
\bm{V}(\si)=(V^{*1},\ldots,V^{*m})=\bm{V}^*.
\]
Hence every optimal equilibrium $\sigma $ achieves the same value.
\end{corollary}

\begin{proof}
(i) is implicit in Definition \ref{defaa03} and (ii) follows immediately from Proposition \ref{propaa02} (iv). Finally (iii) follows immediately from (i) and (ii).
\end{proof}

\subsection{Coalition Values}

The maximin and minimax values, as well as strategy profiles, can be analogously defined for any proper subset $\Aset$ of the set $\Mset$ of players, by simply replacing $k$ with $\Aset$.
For any subset $\Aset \subset \Mset $, we denote $V^\Aset(s):=\sum_{i\in\Aset} V^i(s)$
and we write $s = (s^{\Aset}, s^{-\Aset})$ for any strategy profile $s \in {\cal S}$.

Let us now focus on a special case of a zero-sum game.
It is worth noting that even in zero-sum games, not all Nash equilibria are optimal equilibria.

\begin{definition}\label{defab00}
{\rm A game is called \emph{zero-sum} if $V^\Mset(s) : =\sum_{i\in\Mset} V^i(s)=0$ for all $s\in\Strat$.}
\end{definition}

In the case of a zero-sum game, further properties an optimal equilibrium can be derived.

\begin{proposition}\label{propab01}
Suppose that the game is zero-sum and denote by $\sigma$ any strategy profile. Then the following statements are equivalent:
\\(i) For each $k\in\Mset$, $V^k(\si^k,\si^{-k}) \leq V^k(\si^k,s^{-k})$ for all $s^{-k}\in\Strat^{-k}$.
\\(ii) For any proper subset $\Aset\subset\Mset$, $V^\Aset(\si^{\Aset},\si^{-\Aset}) \geq V^\Aset(s^{\Aset},\si^{-\Aset})$ for all $s^{\Aset}\in\Strat^{\Aset}$.
\\(iii) For any proper subset $\Aset\subset\Mset$, $V^\Aset(\si^{\Aset},\si^{-\Aset}) \leq V^\Aset(\si^{\Aset},s^{-\Aset})$ for all $s^{-\Aset}\in\Strat^{-\Aset}$.
\\(iv) The strategy profile $\si$ is an optimal equilibrium.
\end{proposition}

\begin{proof}\
[(i)$\implies$(ii)]\, Intuitively, if each player $k\in\Aset$ can guarantee his payoff by using $\si^k$, then collectively, the players of $\Aset$ can guarantee their total payoff by playing $\si^\Aset$. Since for all $s^{-\Aset}\in\Strat^{-\Aset}$ we have $(\si^{\Aset\setminus\{k\}},s^{-\Aset})\in\Strat^{-k}$. By (i),
\begin{gather}\label{eqab011}
V^k(\si^\Aset,\si^{-\Aset})=V^k(\si^k,\si^{\Aset\setminus\{k\}},\si^{-\Aset})\leq V^k(\si^k,\si^{\Aset\setminus\{k\}},s^{-\Aset})=V^k(\si^\Aset,s^{-\Aset}).
\end{gather}
Summing \eqref{eqab011} over $k\in\Aset$, we have
\[
V^\Aset(\si^\Aset,\si^{-\Aset})=\sum_{k\in\Aset} V^k(\si^\Aset,\si^{-\Aset})\leq \sum_{k\in\Aset} V^k(\si^\Aset,s^{-\Aset})=V^\Aset(\si^\Aset,s^{-\Aset}),
\]
as required.

\noindent  [(ii)$\implies$(i)]\, This is immediate by setting $\Aset=\{k\}$.

\noindent  [(ii)$\iff$(iii)]\, Since the game is zero-sum, $V^\Aset(s)=-V^{-\Aset}(s)$ for all $s\in\Strat$. Hence
\begin{alignat}{3}\label{eqab012}
&&V^\Aset(\si^\Aset,\si^{-\Aset}) &\leq V^\Aset(\si^\Aset,s^{-\Aset}),\quad &\forall\, s^{-\Aset}&\in\Strat^{-\Aset},\nonumber\\
\iff&&\quad -V^{-\Aset}(\si^\Aset,\si^{-\Aset}) &\leq -V^{-\Aset}(\si^\Aset,s^{-\Aset}),\quad &\forall\, s^{-\Aset}&\in\Strat^{-\Aset},\nonumber\\
\iff&&\quad V^{-\Aset}(\si^\Aset,\si^{-\Aset}) &\geq V^{-\Aset}(\si^\Aset,s^{-\Aset}),\quad &\forall\, s^{-\Aset}&\in\Strat^{-\Aset}.
\end{alignat}
Relabelling $-\Aset$ as $\Aset$ in \eqref{eqab012} gives the desired result.

\noindent  [(i) and (iii)$\implies$(iv)]\,From Definition \ref{defaa02}, the definition of an optimal equilibrium, it is sufficient to check \eqref{eqaa011}
\[
V^k(\si^{k}, \si^{-k})\geq V^k(s^k, \si^{-k}),\quad \forall \, s^k\in\Strat^k
\] and \eqref{eqaa021}
\[
V^k(\si^k, s^{-k})\geq V^k(\si^k, \si^{-k}),\quad \forall \, s^{-k}\in \Strat^{-k}.
\]
This is clear as (i) is identical to \eqref{eqaa011}, while (iii) reduces to \eqref{eqaa021} after setting $\Aset=\{k\}$.

\noindent  [(iv)$\iff$(i)]\, This is immediate from \eqref{eqaa021} in Definition \ref{defaa02}.
\end{proof}

Note that the equivalence between Proposition \ref{propab01} (i) and (ii) does not use the fact that the game is zero-sum. Proposition \ref{propab01} (iii) implies \eqref{eqaa011}, the Nash equilibrium condition, but the converse does not hold.

So far, the definition of the value referred to the value of the game to each individual player. But suppose some subset of players $\Aset\subseteq\Mset$ is playing as a \emph{coalition} with the goal to maximise the \emph{coalition payoff} $V^\Aset(s):=\sum_{i\in\Aset} V^i(s)$. The definition of the value of the game can be extended as follows.

\begin{definition} {\rm
The \emph{value $V^{*\Aset}$ for the coalition} $\Aset$ of players is given by
\[
V^{*\Aset} := \sup_{s^\Aset\in\Strat^\Aset}\inf_{s^{-\Aset}\in\Strat^{-\Aset}} V^\Aset(s^\Aset, s^{-\Aset})=\inf_{s^{-\Aset}\in\Strat^{-\Aset}} \sup_{s^\Aset\in\Strat^\Aset} V^\Aset(s^\Aset, s^{-\Aset}),
\]
assuming that the second equality holds.}
\end{definition}

In general, the value $V^{*\Aset}$ (if well defined) does not necessarily satisfy the \emph{additivity property} $V^{*\Aset}=\sum_{i\in\Aset} V^{*i}$. However, as the following propositions shows,
if an optimal equilibrium exist in a zero-sum game, then this property is indeed satisfied for any subset $\Aset $ of $\Mset$.

\begin{proposition}\label{propab01a1}
Suppose the game has an optimal equilibrium $\si\in\Strat$ with value $V^*=V(\si)$. If either (i) the game is zero-sum or (ii) $V^{\Mset}(\si) = \esssup_{s\in\Strat} V^\Mset(s)$, then for all subsets $\Aset\subseteq\Mset$, we have that
\[
V^\Aset(\si)=\sum_{i\in\Aset} V^{*i}=V^{*\Aset}
= \sup_{s^\Aset\in\Strat^\Aset}\inf_{s^{-\Aset}\in\Strat^{-\Aset}} V^\Aset(s^\Aset, s^{-\Aset})=\inf_{s^{-\Aset}\in\Strat^{-\Aset}} \sup_{s^\Aset\in\Strat^\Aset} V^\Aset(s^\Aset, s^{-\Aset}).
\]
In particular, the value of the game for any coalition $\Aset$ of players is equal to the sum of values for players from~$\Aset $.
\end{proposition}

\begin{proof}
Since case (ii) includes the zero-sum case (i), it suffices to establish the result for case (ii) only.
Since $\si$ is an optimal equilibrium, each player $i\in\Mset$ can guarantee the payoff $V^i(\si)$, in other words,
\[
V^i(\si)=\inf_{s^{-i}\in\Strat^{-i}} V^i(\si^i, s^{-i}).
\]
So the players from $\Aset$ can play to guarantee $V^\Aset(\si)$ by playing $\si^\Aset$,
\[
V^\Aset(\si) = \sum_{i\in\Aset} \inf_{s^{-i}\in\Strat^{-i}} V^i(\si^i, s^{-i}) \leq \sum_{i\in\Aset} V^i(\si^\Aset, s^{-\Aset}) =  V^\Aset(\si^\Aset, s^{-\Aset}), \quad \forall\,s^{-\Aset}\in\Strat^{-\Aset}.
\]
Hence
\begin{gather}\label{eqab01a2}
V^\Aset(\si) \leq\inf_{s^{-\Aset}\in\Strat^{-\Aset}} V^\Aset(\si^\Aset, s^{-\Aset}) \leq \sup_{s^\Aset\in\Strat^\Aset}\inf_{s^{-\Aset}\in\Strat^{-\Aset}} V^\Aset(s^\Aset, s^{-\Aset}).
\end{gather}
Applying the same argument to the player set $-\Aset$, they can also guarantee their payoff $V^{-\Aset}(\si)$,
\[
V^{-\Aset}(\si) \leq V^{-\Aset}(s^\Aset, \si^{-\Aset}), \quad \forall\,s^{\Aset}\in\Strat^{\Aset}.
\]
Now by condition (ii) $V^{\Mset}(\si) = \sup_{s\in\Strat} V^\Mset(s)$, we must have, for all $s^{\Aset}\in\Strat^{\Aset}$,
\begin{align*}
V^\Aset(\si) = V^{\Mset}(\si) - V^{-\Aset}(\si)  \geq V^{\Mset}(s^\Aset, \si^{-\Aset}) -V^{-\Aset}(s^\Aset, \si^{-\Aset})
= V^{\Aset}(s^\Aset, \si^{-\Aset}).
\end{align*}
Hence
\begin{gather}\label{eqab01a3}
V^\Aset(\si) \geq \sup_{s^\Aset\in\Strat^\Aset} V^{\Aset}(s^\Aset, \si^{-\Aset}) \geq \inf_{s^{-\Aset}\in\Strat^{-\Aset}} \sup_{s^\Aset\in\Strat^\Aset} V^\Aset(s^\Aset, s^{-\Aset}).
\end{gather}
Finally, the required result follows immediately by combining \eqref{eqab01a2} and \eqref{eqab01a3} with the following fact (see \ref{propaa02} (i)),
\[
\inf_{s^{-\Aset}\in\Strat^{-\Aset}} \sup_{s^\Aset\in\Strat^\Aset} V^\Aset(s^\Aset, s^{-\Aset}) \geq \sup_{s^\Aset\in\Strat^\Aset}\inf_{s^{-\Aset}\in\Strat^{-\Aset}} V^\Aset(s^\Aset, s^{-\Aset}),
\]
completing the proof.
\end{proof}


\begin{thebibliography}{99} { \parskip =4  pt

\bibitem{Ayache} Ayache, E., Forsyth, P. and Vetzal, K.: Valuation of
convertible bonds with credit risk. \textit{J. Derivatives}, Fall 2003.

\bibitem{Andersen} Andersen, L. and Buffum, L.: Calibration and implementation of
convertible bond models. \textit{J. Comput. Finance} 7 (2004), 1--34.

\bibitem{Bielecki} Bielecki, T. R., Cr\'{e}pey, S., Jeanblanc, M. and Rutkowski, M.:
Arbitrage pricing of defaultable game options with applications to
convertible bonds. \textit{Quant. Finance} 8 (2008), 795--810.

\bibitem{Bielecki1} Bielecki, T. R. and Rutkowski, M.:
Valuation and hedging of contracts with funding costs and collateralization.
Working paper, 2013.

\bibitem{Cvitanic} Cvitani\'c, J. and Karatzas, I.:
Backward stochastic differential equations with reflection and Dynkin games.
\textit{Ann. Probab.} {24} (1996), 2024--2056.

\bibitem{Dolinsky} Dolinsky, Y. and Kifer, Y.:
Hedging with risk for game options in discrete time.
\textit{Stochastics: Int. J. Probab. Stoch. Process.}
79 (2007), 169--195.

\bibitem{Dolinsky11} Dolinsky, Y., Iron, Y. and Kifer, Y.:
Perfect and partial hedging for swing game options in discrete time.
\textit{Math. Finance} 21 (2011), 447--474.

\bibitem{Dynkin} Dynkin, E. B.: Game variant of a problem on optimal stopping.
\textit{Soviet Math. Dokl.} {10} (1969), 270--274.

\bibitem{Follmer} F\"ollmer, H. and Schied, A.:
\textit{Stochastic Finance: An Introduction in Discrete Time.}
2nd ed., De Gruyter, Berlin,  2004.



\bibitem{GuoPhD} Guo, I.:
Competitive multi-player stochastic games with applications to multi-person financial contracts.
Doctoral dissertation, University of Sydney, 2013.

\bibitem{Guo1} Guo, I. and Rutkowski, M.: A zero-sum competitive multi-player game.
\textit{Demonstratio Math.} {45} (2012), 415--433.

\bibitem{Guo2}  Guo, I. and Rutkowski, M.:
Discrete-time multi-player stopping and quitting games with redistribution of payoffs.
In: \emph{Arbitrage, Credit and Informational Risks},
C. Hillairet,  M. Jeanblanc and Y. Jiao, eds., World Scientific, Singapore, 2014, pp. 171--206.

\bibitem{Guo3}  Guo, I. and Rutkowski, M.:
Stochastic multi-player competitive games in discrete time.
Working paper, University of Sydney, 2013.

\bibitem{HH-2011} Hamad\`{e}ne, S. and Hassani, M.:
The multi-player nonzero-sum Dynkin game in continuous time.
Working paper, Universit\'e du Maine, 2011.

\bibitem{HH-2012} Hamad\`{e}ne, S. and Hassani, M.:
The multi-player nonzero-sum Dynkin game in discrete time.
\textit{Math. Meth. Oper. Res.} 79 (2014), 179--194.

\bibitem{Hamadene} Hamad\`{e}ne, S. and Zhang, J.:
The continuous time nonzero-sum Dynkin game problem and application in
game options. \textit{SIAM J. Control Optim.} {48} (2010), 3659--3669.

\bibitem{Jacod} Jacod, J. and Shiryaev, A. N.:
Local martingales and the fundamental asset pricing theorems in the discrete-time case.
\textit{Finance Stoch.} 2 (1998) 259--273.

\bibitem{Jeanblanc} Jeanblanc, M., Yor, M. and Chesney, M.:
{\it Mathematical Methods for Financial Markets.}
Springer, 2009.


\bibitem{Kallenberg} Kallenberg, O.:
\emph{Foundations of Modern Probability.}
Springer, 1997.

\bibitem{Kallsen1} Kallsen, J. and K\"uhn, C.:
Pricing derivatives of American and game type in incomplete
markets. \textit{Finance Stoch.} {8} (2004), 261--284.

\bibitem{Kallsen2} Kallsen, J. and K{\"u}hn, C.:
Convertible bonds: financial derivatives of game type.
In: \emph{Exotic Option Pricing and Advanced L{\'e}vy Models},
Wiley, Chichester New York, 2005, pp. 277--291.


\bibitem{KaratzasLi}  Karatzas, I. and Li, Q.:
BSDE approach to non-zero-sum differential games of control and stopping.
Working paper, Columbia University.

\bibitem{Karatzas1} Karatzas, I. and Shreve, S.:
{\it Methods of Mathematical Finance.}
Springer, 1998.

\bibitem{Kifer1} Kifer, Y.: Game options.
\textit{Finance Stoch.} {4} (2000), 443--463.

\bibitem{Kifer2} Kifer, Y.:
Dynkin games and Israeli options.
\textit{ISRN Probability and Statistics} (2013), ID856458, 17 pages.


\bibitem{Kuehn} K\"uhn, C., Kyprianou, A. E. and van Schaik, K.:
Pricing Israeli options: a pathwise approach.
\textit{Stochastics: Int. J. Probab. Stoch. Process.} 79 (2006), 117--137.

\bibitem{Kyprianou} Kyprianou, A. E.:
Some calculations for Israeli options.
\textit{Finance Stoch.} 8 (2004), 73--86.


\bibitem{Laraki} Laraki, R. and Solan, E.: Equilibrium in two-player non-zero-sum Dynkin games in continuous time.
\textit{Stochastics: Int. J. Probab. Stoch. Process.} 85 (2013), 997--1014.

\bibitem {Nie} Nie, T. and Rutkowski, M.:
Multi-player stopping games with redistribution of payoffs and BSDEs with oblique reflection.
\textit{Stoch. Process. Appl.} 124 (2014), 2672--2698.

\bibitem{Ohtsubo}  Ohtsubo, Y.:
A nonzero-sum extension of Dynkin's stopping problem,
\textit{Math. Oper. Res.} 12 (1987), 277--296.

\bibitem{Pallavicini}
Pallavicini, A., Perini, D. and Brigo, D.:
Funding, collateral and hedging: uncovering the mechanism and the subtleties of funding valuation adjustments.
Working paper, 2012.

\bibitem{Peskir} Peskir, G.: Optimal stopping games and Nash equilibrium.
\textit{Theory Probab. Appl.} {53} (2008), 623--638.

\bibitem{Snell} Snell, J.L.: Applications of martingale system theorems.
\textit{Trans. Amer. Math. Soc.} {73} (1952), 293--312.



}

\end{thebibliography}
\end{document}